\documentclass[11pt]{amsart}

\usepackage{bbm}
\usepackage[lite]{amsrefs}
\usepackage{amssymb}
\usepackage[all,cmtip]{xy}
\usepackage{amsmath}
\usepackage{amsthm}
\usepackage{amsfonts}
\usepackage{hyperref}
\usepackage{enumerate}
\usepackage{physics}
\usepackage[titletoc]{appendix}
\usepackage{fancyhdr}
\usepackage[dvipsnames]{xcolor}
\usepackage{comment}
\usepackage{tikz-cd}
\usepackage{tikzpagenodes}
\usepackage{caption}
\usepackage{float}
\usetikzlibrary[patterns]
\usepackage[margin = 1in]{geometry}
\usepackage{setspace}
\usepackage[T1]{fontenc}

\spacing{1.2}

\allowdisplaybreaks

 \newtheorem{thm}{Theorem}[section]

 \newtheorem{cor}[thm]{Corollary}
 \newtheorem{lem}[thm]{Lemma}
 \newtheorem{prop}[thm]{Proposition}
 
 \newtheorem{defn}[thm]{Definition}
  
  \newtheorem{defn-thm}[thm]{Definition-Theorem}
   
   \newtheorem{example}[thm]{Example}

   \renewcommand\[{\begin{equation}}
\renewcommand\]{\end{equation}}
\numberwithin{equation}{section}

\usepackage{amsxtra}
\usepackage{graphicx}
\usepackage{bbm,newlfont}
\usepackage{amscd}
\usepackage{hyperref}
\usepackage[german,english]{babel}
\usepackage{cancel}
\usepackage{amsmath,amsthm,amssymb}
\usepackage{enumerate}
\usepackage{color}
\usepackage{varwidth}
\usepackage{tasks}
\usepackage{mathrsfs}  
\usepackage{setspace}
\usepackage{collectbox}
\usepackage[utf8]{inputenc} 
\usepackage[T1]{fontenc}


\usepackage{graphicx}

\DeclareFontFamily{U}{mathx}{}
\DeclareFontShape{U}{mathx}{m}{n}{<-> mathx10}{}
\DeclareSymbolFont{mathx}{U}{mathx}{m}{n}
\DeclareMathAccent{\widehat}{0}{mathx}{"70}
\DeclareMathAccent{\widecheck}{0}{mathx}{"71}

\begin{document}
\title[GD realization of SUSY W-algebras]{Gelfand-Dickey Realizations of \\
the supersymmetric classical W-algebras for \\ $\mathfrak{gl}(n+1|n)$ and  $\mathfrak{gl}(n|n)$ }
\author[S.Carpentier]{Sylvain Carpentier$^\dagger$}
\address[S. Carpentier]{Department of Mathematical Sciences, Seoul National University, Gwanak-ro 1, Gwanak-gu, Seoul 08826, Korea}
\email{sylcar@snu.ac.kr}

\author[U.R. Suh]{Uhi Rinn Suh$^{\star}$}
\address[U.R. Suh]{Department of Mathematical Sciences and Research institute of Mathematics, Seoul National University, Gwanak-ro 1, Gwanak-gu, Seoul 08826, Korea}
\email{uhrisu1@snu.ac.kr}

\thanks{$^{\dagger}$This work was supported by BK21 SNU Mathematical Sciences Division}
\thanks{$^{\star}$This work was supported by NRF Grant \#2022R1C1C1008698 and  Creative-Pioneering Researchers Program by Seoul National University
}
\date{}

\begin{abstract}
  In this paper we realize the supersymmetric classical $W$-algebras $\mathcal{W}(\overline{\mathfrak{gl}}(n+1|n))$ and $\mathcal{W}(\overline{\mathfrak{gl}}(n|n))$ as differential algebras generated by the coefficients of a monic superdifferential operator $L$. In the case of $\mathcal{W}(\overline{\mathfrak{gl}}(n|n))$ (resp. $\mathcal{W}(\overline{\mathfrak{gl}}(n+1|n))$) this operator is even (resp. odd). We show that the supersymmetric Poisson vertex algebra bracket on these supersymmetric W-algebras  is the supersymmetric analogue of the quadratic Gelfand-Dickey bracket associated to the operator $L$. Finally, we construct integrable hierarchies of evolutionary Hamiltonian PDEs on both W-algebras. A key observation is that to construct these hierarchies on the algebra $\mathcal{W}(\overline{\mathfrak{gl}}(n+1|n))$ one needs to introduce a new concept of even supersymmetric Poisson vertex algebras.

\end{abstract}

\maketitle

\section{Introduction}
Following a series of discoveries made on the Korteweg-de Vries equation \cites{GGKM67,Lax68,MGK68}, the research on integrable nonlinear partial differential equations took off in the 1970s. In particular, inspired by Peter Lax, Gelfand and Dickey introduced the generalized KdV hierarchies \cite{GD76}. These are infinite families of evolutionary PDEs defined on the differential algebra $\mathcal{V}_N$ generated by the coefficients of a monic differential operator of order $N \geq 2$
$$ L= \partial^N+u_1 \partial^{N-1}+...+u_N \, , \, \, \, \mathcal{V}_N= \mathbb{C}[u_i^{(j)} \, | i=1,..., N \, , j \in \mathbb{Z}_+] \, .$$
The total derivative $\partial$ acts on the generators of $\mathcal{V}_N$ by $\partial(u_i^{(j)})=u_i^{(j+1)}$. 
They showed that the hierarchy of Lax equations

$$ \frac{dL}{dt_l}=[L^{l/N}_+,L] \, , \, \, l \geq 1$$
is bi-Hamiltonian for the so-called first and second Gelfand-Dickey brackets.  The fractional powers of $L$ are elements in $\mathcal{V}_N(\!( \partial^{-1})\!)$ and the residue of such a pseudo-differential operator is the coefficient of $\partial^{-1}$. Finally the subscript $+$ denotes the projection of a pseudo-differential operator onto the subspace of differential operators. The Gelfand-Dickey brackets induce the following two Lie algebra structures on the space of functionals $ \smallint \mathcal{V}_N := \mathcal{V}_N / \partial \mathcal{V}_N$
\begin{equation} \label{eq:adler bracket}
 \begin{split}
       &\big[\smallint a  ,  \smallint b \big]_1  =  \,  \,\,   \int \mathrm{Res} \, \,  \Big( \, (L \frac{\delta a}{\delta L})_+\frac{\delta b}{\delta L}-(\frac{\delta a}{\delta L}L)_+ \frac{\delta b}{\delta L} \, \Big), 
       \\ &\big[\smallint a  ,  \smallint b \big]_2  = \,  \,\, \int \mathrm{Res} \, \, \Big( \, (L \frac{\delta a}{\delta L})_+L\frac{\delta b}{\delta L}-L(\frac{\delta a}{\delta L}L)_+ \frac{\delta b}{\delta L}\, \Big), \, 
\end{split}
\end{equation} 
where we used the notation $\smallint$ for the projection map $\mathcal{V}_N \rightarrow \smallint \mathcal{V}_N $.  The variational derivative of a functional $\smallint a  \in \smallint \mathcal{V}_N$ with respect to $L$ is 
 the  pseudo-differential operator 
\begin{equation} \label{varder}
    \frac{\delta a}{\delta L}:= \sum_{i=1}^N \partial^{i-N-1} \sum_{n \geq 0} (-1)^n (\frac{\partial a}{\partial u_i^{(n)}})^{(n)}= \sum_{i=1}^N \partial^{i-N-1} \frac{\delta a}{\delta u_i}\, .
\end{equation}
For a detailed review of the Hamiltonian PDEs formalism we refer to \cite{KW81}. 
To every functional $\smallint a \in \smallint \mathcal{V}_N $ one can associates two derivations $Y_a$ and $Z_a$ of $\mathcal{V}_N$ commuting with the total derivative $\partial$ and defined on the generators of the differential algebra by 
\begin{equation}
    \begin{split}
        Y_a(L) &=  (L \frac{\delta a}{\delta L}-\frac{\delta a}{\delta L}L)_+\, ,\\
        Z_a(L) &= (L \frac{\delta a}{\delta L})_+L-L(\frac{\delta a}{\delta L}L)_+\, .
    \end{split}
\end{equation}
A series of independent works \cites{GD77, Man78} has proved using lengthy calculations that these assignments are Hamiltonian, i.e. for all $\smallint a \, , \smallint b \in \smallint \mathcal{V}_N $
\begin{equation}
    \begin{split}
        [Y_a,Y_b] &= Y_{ [a,b]_1} \, ,\\
         [Z_a,Z_b] &= Z_{ [a,b]_2}\, .
    \end{split}
\end{equation}
It is clear that $Y_a$ and $Z_a$ are of the form $\mathcal{M}((\frac{\delta a}{\delta u_i})_{i=1,...,N})$ for some $N \times N$ matrix differential operator $\mathcal{M}$. Such a matrix is called a Hamiltonian operator. This concept was recently reformulated using the language of $\lambda$-brackets \cite{BDSK09}. More precisely, one can construct a map $\{ \, {}_{\lambda} \, \} : \, \,  \mathcal{V}_N \otimes \mathcal{V}_N \rightarrow \mathcal{V}_N[\lambda]$ such that  
\begin{equation} 
\begin{split}
\{ u_i {}_{\lambda} u_j \} &= \mathcal{M}_{ij}(\lambda) \,  \\
\end{split}
\end{equation}
for all $1 \leq i,j \leq N$ and such that it satisfies all the axioms of a Poisson vertex algebra (PVA). Above, the notation $P(\lambda)$ for a differential operator $P$ simply means replacing $\partial^l$ by $\lambda^l$. Conversely every PVA structure on $\mathcal{V}_N$ uniquely defines an $N \times N$ Hamiltonian matrix differential operator. Even though both notions are equivalent the language of $\lambda$-brackets and PVAs has proven very efficient from a technical point of view and we will follow this approach in our paper.  
\\

In \cite{DS85}, Drinfeld and Sokolov managed to connect these hierarchies of Hamiltonian PDEs to the theory of Kac-Moody Lie algebras. First  they considered a so-called affine Hamiltonian structure $\{ \, \, {}_\lambda \, \, \}_{\text{aff}}$ on the differential algebra $\mathcal{V}(\mathfrak{g})$ of differential polynomials on a  semisimple Lie algebra  $\mathfrak{g}$ defined on its generators by  
\begin{equation} \label{eq:affine H}
    \{g{}_\lambda h\}_{\text{aff}} =[g,h]+\lambda(g|h) \, , \, \, g, h \, \in \mathfrak{g}
\end{equation}
where $( \, | \, )$ is an invariant symmetric bilinear form. They also constructed a hierarchy of compatible equations based on a linear matrix operator $\mathcal{L}$ which are Hamiltonian for the affine structure. Their key result is showing the existence of a differential subalgebra $\mathcal{W}(\mathfrak{g}) \subset \mathcal{V}(\mathfrak{g})$ such that 
\begin{enumerate}
    \item [$1$.] The PVA structure on $\mathcal{W}(\mathfrak{sl}(N+1))$ coincides with the second Gelfand-Dickey Hamiltonian structure, 
    \item[$2$.] The restriction of the affine hierarchy to $\mathcal{W}(\mathfrak{sl}(N+1))$ is the $N$-th KdV hierarchy, 
    \item[$3$.] The algebra $\mathcal{W}(\mathfrak{g})$ coincides with the functions on $\mathcal{L}$ invariant under a gauge action of a nilpotent subalgebra $\mathfrak{n} \subset \mathfrak{g}$.
\end{enumerate}
This process is now known as the Drinfeld-Sokolov reduction and the algebras $\mathcal{W}(\mathfrak{g})$ as the principal classical W-algebras. Moreover, Drinfeld and Sokolov showed in \cite{DS85} that their hierarchies are bi-Hamiltonian where the second PVA structure $\{ \, \, {}_\lambda \, \, \}_{\text{lin}}$ is defined on $\mathcal{V}(\mathfrak{g})$  by 
\begin{equation} \label{eq:affine K}
    \{g{}_\lambda h\}_{\text{lin}} =(s|[g,h])
\end{equation}
for $g,h\in \mathfrak{g}$. Here $s\in \mathfrak{g}$ is an element of highest principal grading. After reduction to the subalgebra $\mathcal{W}(\mathfrak{g})$ this linear bracket coincides with the first Gelfand-Dickey bracket in the case where $\mathfrak{g}=\mathfrak{sl}(N+1)$.
\\

Later, the notion of W-algebra was further generalized and the classical W-algebra $\mathcal{W}(\mathfrak{g},F)$ was constructed 
for any given basic simple Lie superalgebra $\mathfrak{g}$ and even nilpotent element $F\in \mathfrak{g}$ \cite{KRW04}. Note that when $\mathfrak{g}$ is a Lie algebra and $F$ is principal, the corresponding W-algebra  $\mathcal{W}(\mathfrak{g},F)$ is just $\mathcal{W}(\mathfrak{g})$. Unfortunately, one can only construct integrable hierarchies using the method of Drinfeld-Sokolov for a few classical W-algebras.  De Sole-Kac-Valeri successfully overcame that difficulty for arbitrary simple Lie algebra $\mathfrak{g}$ by introducing matrix valued Adler operators \cites{DSKV15, DSKV16a, DSKV16b, DSKV18}. When $\mathfrak{g}$ is a Lie superalgebra, there are only partial results as far as integrable hierarchies are concerned \cites{Suh18, CLS24}.
\\

Our overarching goal is to generalize this picture to classical supersymmetric (SUSY) W-algebras.  SUSY W-algebras were introduced by Madsen-Ragoucy \cite{MR94} and later reinterpretated by Molev-Ragoucy-Suh \cite{MRS21} using the language of SUSY vertex algebras \cite{K}. Roughly speaking a SUSY vertex algebra is a vertex algebra together with an odd derivation $D$ whose square is the even derivation $\partial$ of the vertex algebra. The OPE relation of SUSY vertex algebras can be expressed using a SUSY analogue of the $\lambda$-bracket \cite{HK07} where the even indeterminate $\lambda$ is replaced by a formal odd indeterminate $\chi$ such that $ D \chi+\chi D=-2 \chi^2 \, . $
\\

 As a fundamental example, the OPE product of the affine SUSY vertex algebra $V^k(\bar{\mathfrak{g}})$ associated with a Lie superalgebra $\mathfrak{g}$ is determined by the odd bracket 
\begin{equation}
    [\bar{g}{}_{\chi}\bar{h}]= \overline{[g,h]}+k\chi(g|h)
\end{equation}
for $g,h\in \mathfrak{g}$ and $k\in \mathbb{C}$. Here $\bar{\mathfrak{g}}$ is the parity reversed superspace of $\mathfrak{g}.$ The classical SUSY W-algebra $\mathcal{W}^k(\bar{\mathfrak{g}},f)$ of level $k$ associated with a basic Lie superalgebra $\mathfrak{g}$ and an odd nilpotent $f$ is defined as a Hamiltonian reduction of the quasi-classical limit of $V^k(\bar{\mathfrak{g}}).$  They are isomorphic whenever $k \neq 0$ and we will simply write $\mathcal{W}(\bar{\mathfrak{g}},f):=\mathcal{W}^1(\bar{\mathfrak{g}},f)$. The SUSY PVA structure of $\mathcal{W}(\bar{\mathfrak{g}},f)$ that we denote by $\{ \cdot {}_\chi \cdot \}^q$ is well-known especially when $f$ is a odd principal nilpotent.
In this paper we start our project of constructing integrable systems for each of these algebras in the case where $\mathfrak{g}=\mathfrak{gl}(m|n)$ for $m=n$ or $n+1$ and the odd nilpotent element $f$ is principal \cites{MRS21,RSS23,LSS23,Suh20}. In the spirit of Drinfeld and Sokolov, we show that for these types the reduction of the affine brackets leads to SUSY analogues of the Gelfand-Dickey brackets \cite{HN91}. 
 \\
 
 Let $N$ be an integer greater than $2$. Consider the differential superalgebra $\mathcal{W}_N$ generated by the coefficients of the super-differential operator 
 $$L=D^N+u_1D^{N-1}+...+u_N \, , \quad \, p(u_i) \equiv i \ (\text{mod} 2) \, .$$
The odd derivation $D$ acts on the generators of $\mathcal{W}_N$ by $D(u_i^{(j)})=u_i^{(j+1)}$. The following hierarchy of equations
\begin{equation} \label{susylax}
\frac{dL}{dt_l}=[L^{2l/N}_+ \, , \, L] \, ,  \quad  \, [d/dt_l \, , \, D] = 0 \, , \, \, l \geq 1 ,
\end{equation}
defines pairwise commuting even evolutionary derivations of $\mathcal{W}_N$. However, there is a fundamental difference depending on the parity of $N$:
\begin{enumerate}[]
    \item 1. If $N=2n$ is even, the conserved densities  $(\text{res } L^{\frac{q}{n}})_{q \geq 1}$ of the hierarchy \eqref{susylax} are \textbf{odd}.  
    \item 2. If $N=2n+1$ is odd, the conserved densities  $(\text{res } L^{\frac{2q-1}{2n+1}})_{q \geq 1}$ of the hierarchy \eqref{susylax} are \textbf{even}.
\end{enumerate}
We note that in the first case the $2n$-th root of $L$ is not defined whilst in the second case the residues of even powers of $L^{\frac{1}{2n+1}}$ are trivial. We want to discuss the Hamiltonian formalism underneath these hierarchies. These Hamiltonian structures should map conserved densities to even derivations. As we just noted it must be of different nature depending on the parity of $N$. In this paper, we construct these and relate them to the SUSY Gelfand-Dickey brackets as well as some SUSY classical W-algebras.
\\

What is commonly understood by SUSY Hamiltonian operator? An evolutionary derivation of $\mathcal{W}_N$ is a derivation which supercommutes with the total derivative $D$ hence which is completly determined by its values on the generators $u_i$. We denote such an object by $X_F$ where $F$ is the characteristic vector $X_F(u_i)=F_i$. We send the reader to the main text \eqref{Xvder} for the definition of the variational derivative $ \delta / \delta L : \smallint \mathcal{W}_N \rightarrow (\mathcal{W}_N)^N$. An $N \times N$ matrix differential operator $H$ is called Hamiltonian if the subspace $\{ X_{H(\frac{\delta a}{\delta L})} \, | \, \smallint a \in \smallint \mathcal{W}_N \}$ is a sub Lie superalgebra of the superalgebra of evolutionary derivations of $\mathcal{W}_N$. Essentially, this matrix maps functionals to evolutionary derivations. If the map is odd (resp. even) we say that it is an odd (resp. even) Hamiltonian operator.
\\

Taking an odd Hamiltonian matrix differential operator $H$ and replacing $D$ by the  odd indeterminate $\chi$, one obtains a $\chi$-bracket $\{ \,  \cdot \, {}_\chi \, \cdot\,  \}: \mathcal{W}_N \otimes \mathcal{W}_N \rightarrow \mathbb{C}[\chi]\otimes \mathcal{W}_N$
satisfying the
 axioms of a SUSY PVA  \cite{HK07}. This bracket is odd, i.e. for all $a, b \in \mathcal{W}_N$ we have 
 $$p(\{ a \, {}_{\chi} \, b\})=p(a)+p(b)+1. $$ 
 Moreover, the Hamiltonian derivation can be rewritten as 
 $$ X_{H(\frac{\delta a}{\delta L})}= \{  \, a \, {}_{\chi} \, \cdot \, \}  \big| _{\chi=0} \, .$$

\begin{thm}[Theorem \ref{thm:main1}] \label{thm:intro_main1}
    Let $m$ and $n$ be positive integers such that $m=n$ or $m=n+1$ and $N=m+n.$ Consider the Lie superalgebra $\mathfrak{g}=\mathfrak{gl}(m|n)$ and the principal nilpotent element $f$.
    There is 
    a differential algebra isomorphism $\phi$ from the SUSY W-algebra $\mathcal{W}(\bar{\mathfrak{g}},f)$ to $\mathcal{W}_{N}$, such that the SUSY PVA structure of $\mathcal{W}(\bar{\mathfrak{g}},f)$  induces the so-called SUSY quadratic Gelfand-Dickey bracket on $\int \mathcal{W}_N$. More precisely, for any $a,  b \in \mathcal{W}_N$,
    \begin{equation}
       \int \phi \{ \phi^{-1}(a) \, {}_{\chi} \, \phi^{-1}(b)\}^q |_{\chi=0} = (-1)^{N+p(a)} \int \mathrm{Res} \, \, \,  ((L\frac{\delta a}{\delta L})_+L\frac{\delta b}{\delta L}-L(\frac{\delta a}{\delta L} L)_+\frac{\delta b}{\delta L}).
    \end{equation}
\end{thm}
\vskip 5mm

It is known \cite{HK07} that a SUSY PVA bracket on a differential superalgebra $\mathcal{P}$ induces a degree $1$ Lie superalgebra bracket $\int \{ \cdot {}_\chi {} \cdot \} |_{\chi=0}$ on the quotient $\mathcal{P}/D\mathcal{P}.$ Hence the following corollary directly follows.

\begin{cor} The SUSY quadratic Gelfand-Dickey bracket in Definition \ref{def:quad_GD} is a degree $1$ Lie superalgebra bracket on the space of functionals $\int \mathcal{W}_N.$
    
\end{cor}
This bracket was introduced in \cite{HN91} where it is claimed to be a Poisson bracket based on the work \cite{STS83}, but we have not found a explicit proof of the fact that it satisfies the Jacobi identity.
\begin{thm} [Theorem \ref{thm:integrability_even}]
    The hierarchy \eqref{susylax} is Hamiltonian for the SUSY quadratic Gelfand-Dickey structure when $N=2n$ is even with conserved densities $\text{ Res } L^{q/n}$ for $q \geq 1$. 
\end{thm}
However, when $N=2n+1$ is odd, we must construct an even Hamiltonian operator $K$ to connect the even conserved densities to the even evolutionary derivations \eqref{susylax}. The first example of even Hamiltonian operator appeared in \cite{Pop09} in the study of the supersymmetric Sawada-Kotera equation. We note that in \cite{Pop09} this operator is called odd since it induces a odd Poisson bracket in a different sense.  This leads us to define in section \ref{sec:even PVA} the concept of even SUSY PVAs. This structure is given by a even bracket $\mathcal{W}_{2n+1} \otimes \mathcal{W}_{2n+1} \rightarrow \mathbb{C}[\chi]\otimes\mathcal{W}_{2n+1}$ such that for all $a,b \in \mathcal{W}_{2n+1}$, 
 $$p(\{ a \, {}_{\chi} \, b \})=p(a)+p(b). $$  
In particular, we have for all $a \in \mathcal{W}_{2n+1}$
 $$X_{K(\frac{\delta a}{\delta L})}= \{ a \, {}_{\chi} \, \cdot \}|_{\chi=0} \, . $$
We then construct even SUSY PVA brackets on the differential algebras $\mathcal{V}^k(\overline{\mathfrak{gl}}(n+1|n))$. We show that the SUSY W-algebra $\mathcal{W}(\overline{\mathfrak{gl}}(n+1|n),f)$ is stable for this even SUSY PVA bracket $\{ \cdot \, {}_{\chi} \, \cdot \}^{e}$ and moreover prove the following theorem.
\begin{thm} [Theorem \ref{thm:even bracket W_N}]
    Recall the map $\phi$ in Theorem \ref{thm:intro_main1}. The image by $\phi$ of the reduction of the even SUSY PVA structure to the space of functionals $\smallint \mathcal{W}_{2n+1}$ is a Lie superalgebra bracket of Gelfand-Dickey type. More precisely, for any $ a,  b \in \mathcal{W}_{2n+1}$,
    \begin{equation}
       \int \phi \{ \phi^{-1}(a) \, {}_{\chi} \, \phi^{-1}(b)\}^{e} = 
      \int \, {\{ a \, {}_{\chi} \, b \}^{e}}|_{ \chi = 0} =  \int \text{ Res } (L \frac{\delta a}{\delta L} \frac{\delta b}{\delta L} +(-1)^{p(a)+1} \frac{\delta a}{\delta L} L \frac{\delta b}{\delta L}).  
    \end{equation}
\end{thm}
Finally, we conclude section \ref{sec:even PVA} by
\begin{thm}[Theorem \ref{thm:integrability odd}]
    The hierarchy \eqref{susylax} for odd $N=2n+1$ is Hamiltonian for the even linear SUSY Gelfand-Dickey structure in Definition \ref{def:el_GD} with conserved densities $\text{ Res } L^{\frac{2q-1}{2n+1}}$ for $q \geq 1$. 
\end{thm}
We suggest that the structures now referred to as SUSY PVAs in the literature should be called odd SUSY PVAs when the context is ambiguous, to distinguish them from the even SUSY PVAs introduced in this paper.

\section{supersymmetric classical
W-algebras} \label{sec:SUSY W-alg}
\subsection{SUSY Poisson vertex algebra}
\label{subsec:SUSY PVA}
Throughout this paper, we consider vector superspaces defined over $\mathbb{C}$. For such a vector superspace $V$, 
we denote its even (resp. odd) subspace  by $V_{\bar{0}}$ (resp. $V_{\bar{1}}$)
and the parity of a homogeneous element $a\in V$ by $p(a).$

Let $\mathcal{R}$ be a $\mathbb{C}[D]$-module, where $D$ is an odd operator acting on a superspace $\mathcal{R}$. For an odd indeterminate $\chi$, we extend the $\mathbb{C}[D]$-module structure to 
$\mathbb{C}[\chi]\otimes R$ by letting
$ D\chi + \chi D =- 2 \chi^2.$
A linear map 
    \[ [ \cdot\,  {}_\chi\,  \cdot  ] : \mathcal{R}\otimes \mathcal{R} \to \mathbb{C}[\chi] \otimes \mathcal{R}\]
is called a (odd) $\chi$-bracket on $\mathcal{R}$ if for all $a, b \in R$
\begin{enumerate}
\item \textrm{(odd parity)} $p([a{}_\chi b] )= p(a)+p(b)+1$,
    \item \textrm{(sesquilinearity)} $[Da{}_\chi b]=\chi [a{}_\chi b],$ $[a{}_\chi Db]= (-1)^{p(a)+1} (D+\chi) [a_\chi b] \, .$
\end{enumerate}
 Moreover, if a $\chi$-bracket satisfies the properties:
\begin{enumerate}
    \item \textrm{(skew-symmetry)} $[b{}_\chi a]=(-1)^{p(a)p(b)}[a{}_{-\chi-D} b]_{\leftarrow},$
    \item \textrm{(Jacobi identity)} $[a{}_{\chi} [b{}_{\gamma} c]]= (-1)^{p(a)+1}[[a{}_{\chi} b]_{\chi+\gamma} c] + (-1)^{(p(a)+1)(p(b)+1)}[b_{\gamma} [a_{\chi} c]]$
\end{enumerate}
for all $a,b,c\in \mathcal{R}$,
then $\mathcal{R}$ is called a \textit{supersymmetric Lie conformal algebra} (SUSY LCA). We denote by $a_{(n)}b\in \mathcal{R}$ the coefficient of $\chi^n$ in $[a{}_\chi b]$, so that $[a{}_\chi b]=\sum_{n\in \mathbb{Z}_+} \chi^n a_{(n)}b.$  In the RHS of the skew-symmetry axiom, $[a_{-\chi-D}b]_{\leftarrow}:=\sum_{n\in \mathbb{Z}_+} (-D-\chi)^n a_{(n)}b$ and the Jacobi identity holds in $\mathbb{C}[\chi,\gamma]\otimes \mathcal{R}$ for an odd indeterminate $\gamma$ such that $\chi \gamma + \gamma \chi=0$ and $D\gamma+\gamma D =-2 \gamma^2.$ More precisely, the Jacobi identity is computed by the following formula.
\begin{equation}
\begin{aligned}
     & [ a{}_{\chi}\gamma^n b]=(-1)^{n}\gamma^n [a{}_{\chi } b], \quad [a{}_{\gamma} \chi^n b]= (-1)^n \chi^n [a{}_\chi b], \\
     & [\chi^n a{}_{\chi+\gamma}b]=(-1)^n \chi^n [a {}_{\chi+\gamma}b].
\end{aligned}
\end{equation}

\begin{defn} \label{Def:SUSY PVA} Let $\mathcal{P}$ be a unital supersymmetric differential algebra with an odd derivation D.
    A SUSY Lie conformal algebra $(\mathcal{P}, D, \{ \ {}_\chi \ \})$ is called an odd SUSY Poisson vertex algebra (odd SUSY PVA) if it satisfies the Leibniz rule:
      \[ \{ a{}_{\chi} bc \}= \{a{}_{\chi} b\}c+(-1)^{(p(a)+1)p(b)}b\{a{}_{\chi}c\}\quad  \text{ for } \quad  a,b,c\in \mathcal{P}.\]
\end{defn}

Note that this definition of an odd SUSY PVA corresponds to the $N_k=1$ SUSY PVA in \cite{HK07} In this paper, we are interested in SUSY PVAs which are algebras of differential polynomials, i.e. of the form
\begin{equation} \label{eq:P}
     \mathcal{P}= \mathbb{C}[u_i^{(m)}|i\in I, m\in \mathbb{Z}_+],
\end{equation}
where $I$ is a finite set and the odd derivation $D$ acts on the generators of $P$ by $D(u_i^{(m)}):= u_i^{(m+1)}$.

\begin{prop}  \cite{CS2020}\label{prop:LCA to PVA}
Let $\mathcal{P}$ be a differential algebra of polynomials \eqref{eq:P} endowed with a $\chi$-bracket $\{\ {}_\chi\ \}$.  Suppose that the skew-symmetry and Jacobi identity axioms hold for a set of generators $(u_i)_{i \in I}$ of $\mathcal{P}$, i.e., 
    \begin{equation*}
    \begin{aligned}
    & \{u_i\, {}_\chi\,  u_j\}=(-1)^{\tilde{\imath} \tilde{\jmath}} \{u_j\, {}_{-\chi-D}\, u_i\}_{\leftarrow};\\
    & \{u_i\, {}_\chi\, \{ u_j \,{}_\gamma u_k \, \}\}=(-1)^{\tilde{\imath}+1}\{\{u_i\, {}_\chi \, u_j\}{}_{\chi+\gamma} u_k\}+ (-1)^{(\tilde{\imath}+1)(\tilde{\jmath}+1)} \{u_j\, {}_\gamma \{u_i \, {}_\chi \, u_k\}\}
    \end{aligned}
    \end{equation*}
    for $i,j,k\in I$. Here $\tilde{\imath}$ and $\tilde{\jmath}$ are parities of $u_i$ and $u_j$, respectively. If the bracket $\{\, {}_\chi \, \}$ also satisfies the Leibniz rule, then $P$ is a SUSY PVA.
   
    \end{prop}
\begin{proof}
    See \cite{CS2020} Theorem 2.15.
\end{proof}

Thanks to the Leibniz rule, sesquilinearity and skew-symmetry axioms, an odd SUSY PVA bracket on a algebra of differential polynomials is completely determined by its values on pairs of generators $(u_i,u_j), \, i, j \in I.$ We will now try to find a formula expressing the bracket $\{ a \, {}_{\chi} \, b \}$ for all $a, b \in \mathcal{P}$ in terms of the brackets $ \{ u_i \, {}_{\chi} \, u_j \}$. For any $a , b \in \mathcal{P}$, we define the super-differential operator 
$$  \{a \, {}_{D} \, b\}\, :=\sum_{n \in \mathbb{Z}_+ }(-1)^{n(p(a)+p(b))+\frac{n(n-1)}{2}} a_{(n)}b \, D^n  \, \, .$$
Conversely, we can retrieve the $\chi$-bracket of $a$  and $b$ from the operator $\{a \, {}_{D} \, b\}$ by the formula 
$$ \{a \, {}_{D+ \chi} \, b\}(1)=\{a \, {}_{\chi} \, b\} \, .$$
For any $a,b,c \in \mathcal{P}$ we will make the slight abuse of notation
$$ \{a \, {}_{D + \chi} \, b\}\, \, c :=\sum_{n \in \mathbb{Z}_+ }(-1)^{n(p(a)+p(b))+\frac{n(n-1)}{2}} a_{(n)}b \, (D + \chi)^n (c) \, . $$
By the skew-symmetry and Leibniz rule, we have 
\begin{equation}
    \{ab{}_\chi c\}=(-1)^{p(b)p(c)}\{a{}_{\chi+D} c\} b + (-1)^{p(a)(p(b)+p(c))}\{b{}_{\chi+D} c\}a
\end{equation} 
and hence for all $a,b\in \mathcal{P}$,
    \begin{equation} \label{eq:var-deriv-1}
    \begin{aligned}
        \{\, a\, {}_\chi\, b\, \}& =\sum_{i\in I, m\in \mathbb{Z_+}}(-1)^{(p(a)+\tilde{\imath}+m) p(b)} \{u_i^{(m)}{}_{\chi+D} \, b\} \frac{\partial a}{ \partial u_i^{(m)}}  \\
        & =\sum_{i\in I, \, m\in \mathbb{Z}_+} (-1)^{(p(a)+\tilde{\imath}+m)p(b)+m(\tilde{\imath}+p(b)+1)+ \frac{m(m-1)}{2}}\{u_i \, {}_{\chi+D} \, b\} (\chi+D)^m \frac{\partial a}{\partial u_i^{(m)}}.
    \end{aligned}
    \end{equation}
Finally we obtain what we will call the \textit{master formula}:
    \begin{equation} \label{eq:master}
        \begin{aligned}
            \{\, a\, {}_\chi\, b\, \}& = \sum_{i,j\in I, \, m,n\in \mathbb{Z}_+} s_{i,j}^{m,n}\frac{\partial b}{\partial u_j^{(n)}} \{\, u_i^{(m)}\, {}_{\chi+D} \, u_j^{(n)} \} \frac{\partial a}{ \partial u_i^{(m)}}\\
            & =  \sum_{i,j\in I, \, m,n\in \mathbb{Z}_+} s_{i,j}^{m,n} \, t_{i,j}^{m,n}\frac{\partial b}{\partial u_j^{(n)}} (\chi+D)^n\{\, u_i\, {}_{\chi+D} \, u_j \}(\chi+D)^m \frac{\partial a}{ \partial u_i^{(m)}},
        \end{aligned}
    \end{equation}
    where $s_{i,j}^{m,n}=(-1)^{p(b)(p(a)+\tilde{\jmath}+n+1)+(\tilde{\jmath}+n)(\tilde{\imath}+m)}$ and $t_{i,j}^{m,n}= (-1)^{n(\tilde{\imath}+m+1)+m(\tilde{\imath}+\tilde{\jmath}+1)+ \frac{m(m-1)}{2}}$ for  $\tilde{\imath}:=p(u_i)$ and   $\tilde{\jmath}:=p(u_j)$. 

\begin{defn}
    Let  $\mathcal{P}$ be the algebra of differential polynomials \eqref{eq:P}. The \textit{variational derivative} of an element $a \in \mathcal{P}$ with respect to $u_i$ is given by
\begin{equation} \label{eq:var-deriv}
  \frac{\delta a}{\delta u_i}= \sum_{m\in \mathbb{Z}_+}(-1)^{m \tilde{\imath} + \frac{m(m+1)}{2}} D^m \frac{\partial a}{\partial u_i^{(m)}}.
\end{equation}
\end{defn}

Let us consider the projection map $\int: \mathcal{P} \to \smallint \mathcal{P} := \mathcal{P}/D(\mathcal{P})$. The notation is justified by its `integration by parts' property $\int D(a)b=- \int (-1)^{p(a)}a D(b)$. We call an element of $\smallint \mathcal{P}$ a \textit{ functional}. By sesquilinearity, given $\smallint a \in \smallint \mathcal{P}$ the following linear map is well-defined 
\begin{equation}
     \mathcal{P} \to \mathcal{P}, \quad v \mapsto \{\smallint a\, ,  \,  v\}:= \{\, a\, {}_\chi \, v\, \}|_{\chi=0}.
\end{equation}
  By the Leibniz rule, this map is a derivation of $\mathcal{P}$ of parity $p(a)+1.$   It follows from the Jacobi identity that the superspace of functionals is endowed with a {\it degree 1 Lie superalgebra bracket} defined by 
\[ \{ \, \smallint a \, , {} \, \smallint b \, \}:= \smallint \{  \, a \, {}_\chi \, b\, \}|_{\chi=0}. \]
Comparing \eqref{eq:master} and \eqref{eq:var-deriv}, we have  
\begin{equation} \label{eq:Lie str}
\begin{split}
    \{\smallint a, \smallint b\} &= \sum_{i,j\in I} (-1)^{p(a)p(b)+p(b)\tilde{\jmath}+p(b)+ \tilde{\imath}\tilde{\jmath}} \int  \frac{\delta b}{\delta u_j} \{u_i{}_D u_j\} ( \frac{\delta a}{\delta u_i} ) \\
        &= \sum_{i,j\in I} (-1)^{p(a)\tilde{\jmath}+\tilde{\imath}\tilde{\jmath}} \int   \{u_i{}_D u_j\} ( \frac{\delta a}{\delta u_i} ) \, \, \frac{\delta b}{\delta u_j} \, .
\end{split}
\end{equation}

\subsection{SUSY classical W-algebras} \label{subsec: SUSY
W-algebra}
Let $\mathfrak{g}$ be a finite basic Lie superalgebra with a $\mathbb{Z}/2$-grading  $\mathfrak{g}=\bigoplus_{i\in \frac{\mathbb{Z}}{2}} \mathfrak{g}(i)$ and denote 
\[ \mathfrak{n}= \bigoplus_{i>0} \mathfrak{g}(i), \quad   \mathfrak{b} = \bigoplus_{i\leq 0 } \mathfrak{g}(i), \quad  \mathfrak{b}_+ = \bigoplus_{i\geq 0 } \mathfrak{g}(i).\]
Suppose that there is an odd element  $f\in \mathfrak{g}(-\frac{1}{2})$ satisfying the properties:
\begin{equation} \label{eq:odd good}
    \text{$\text{ad} f |_\mathfrak{n} $ :   injective\ ,    \quad $\text{ad} f|_{\mathfrak{b}}$ :  surjective}.
\end{equation}
For example, \eqref{eq:odd good} holds when $f$ is an odd nilpotent element in a subalgebra isomorphic to $\mathfrak{osp}(1|2)$ and one takes the Dynkin grading induced from the $\mathfrak{sl}_2$-triple $(E,H,F=\frac{1}{2}[f,f])$.

Let $\bar{\mathfrak{g}}$ be the parity reversed superspace of $\mathfrak{g}.$  The affine SUSY PVA of level $k \in \mathbb{C}$, denoted by 
$\mathcal{V}^{k}(\bar{\mathfrak{g}})$, is isomorphic as a differential algebra to the supersymmetric algebra of the superspace $\mathbb{C}[D]\otimes \bar{\mathfrak{g}}$. The SUSY PVA  $\chi$-bracket  \cite{HK07} on $\mathcal{V}^{k}(\bar{\mathfrak{g}})$ is defined on its generators by
\[ \{\bar{g}{}_\chi \bar{h}\}_{\text{aff}} = (-1)^{p(g)}(\, \overline{[g,h]} +k\chi(g|h) \, )\] for $g,h\in \mathfrak{g}$. The classical SUSY W-algebra associated with $\mathfrak{g}$ and $f$ is defined by 
\begin{equation}
\mathcal{W}^k(\bar{\mathfrak{g}}, f)= \{ w \in \mathcal{V}^{k}(\bar{\mathfrak{g}}) \,  | \, \rho( \{\bar{v}{}_\chi w\}_{\text{aff}} )=0 \text{ for all } v \in \mathfrak{n}\},
\end{equation}
where $\rho : \mathbb{C}[\chi]\otimes \mathcal{V}^{k}(\bar{\mathfrak{g}}) \to \mathbb{C}[\chi]\otimes S(\mathbb{C}[D]\otimes \bar{\mathfrak{b}})$ is the algebra homomorphism defined by
\[D^m\bar{g }\mapsto \left\{\begin{array}{ll} D^m\bar{g} & \text{ if } g\in \mathfrak{b}, \\ -D^m(f|g)  & \text{ if } g \in \mathfrak{n};\end{array} \right.  \quad \chi \mapsto \chi \, ,\]
for $m\in \mathbb{Z}_+$.
One can check that the $\chi$-bracket  on $\mathcal{W}^k(\bar{\mathfrak{g}}, f)$ defined by 
$
 \{w_1{}_\chi w_2\}:=\rho( \{w_1 {}_\chi w_2\}_{\text{aff}})$
 for  $w_1, w_2 \in\mathcal{W}^k(\bar{\mathfrak{g}}, f)$ induces a SUSY PVA structure \cite{Suh20}.
 They are isomorphic whenever $k \neq 0$ and we will simply denote $\mathcal{W}(\bar{\mathfrak{g}},f):=\mathcal{W}^1(\bar{\mathfrak{g}},f)$. By the same reason, we denote $\mathcal{V}(\bar{\mathfrak{g}}):= \mathcal{V}^1(\bar{\mathfrak{g}}).$
 
 Here we remark that in other references, for example \cite{MRS21}, a SUSY W-algebra is defined under the assumption that $f$ is an odd nilpotent element in an $\mathfrak{osp}(1|2)$ subalgebra but \eqref{eq:odd good} is enough for the properties that are used in this paper. 
 
\begin{example} \label{ex:gl(n+1|n)} Let $\mathfrak{g}= \mathfrak{gl}(n+1|n)$. Take a basis $\{\, e_{ij}\, |\, i,j=1,2,\cdots, 2n+1 \, \}$ of $\mathfrak{g}$,
 where the parity of $e_{ij}\in \mathfrak{g}$ is $p(e_{ij})\equiv i+j\  (\text{mod }2)$ and  $e_{ij} \in \mathfrak{g}(\frac{j-i}{2})$. One can check that the odd element 
 $f=\sum_{i=1}^{2n} e_{i+1\, i} \in \mathfrak{g}(-\frac{1}{2})$
is in a subalgebra of $\mathfrak{g}$ isomorphic to $\mathfrak{osp}(1|2)$. Hence $\mathcal{W}(\bar{\mathfrak{g}},f)$ is well-defined.
 \end{example}

\begin{example} \label{ex:gl(n|n)}
    Let $\mathfrak{g}=\mathfrak{gl}(n|n)$ and take a basis $\{\, e_{ij}\, |\, i,j=1,2,\cdots, 2n \, \}$, where the parity of $e_{ij}\in \mathfrak{g}$ is $p(e_{ij})\equiv i+j\  (\text{mod }2)$. Consider the grading on $\mathfrak{g}$ such that $e_{ij}\in \mathfrak{g}(\frac{j-i}{2})$ and take an odd element $f=\sum_{i=1}^{2n-1} e_{i+1\, i} \in \mathfrak{g}(-\frac{1}{2})$. Then $f$ satisfies the properties \eqref{eq:odd good} and  hence $\mathcal{W}(\bar{\mathfrak{g}},f)$ is well-defined.
\end{example}

 \subsection{Lax matrix operators} Let $\bar{\mathfrak{b}}$ be the parity reversed superspace of $\mathfrak{b}$ and $\mathcal{V}(\bar{\mathfrak{b}})$ be the differential superalgebra of polynomials generated by $\bar{\mathfrak{b}}$.
Let $\{q^i|i\in S\}$ and $\{ {q}_i| i\in S\}$ be bases of $\mathfrak{b}_+$ and ${\mathfrak{b}}$ such that $(q^i |q_j) = \delta_{ij}$.
We introduce the {\it universal Lax operator }
 \begin{equation} \label{eq: universal L}
  \mathcal{L}^u = D + \sum_{i\in S} q^{i} \otimes \bar{q}_i - f\otimes 1 \in (\mathbb{C}D\oplus \mathbb{C} \partial) \ltimes \left(\mathfrak{g} \otimes \mathcal{V}(\bar{\mathfrak{b}}) \right),
  \end{equation}
 where $\mathbb{C}[D]$ acts on $\mathfrak{g} \otimes \mathcal{V}(\bar{\mathfrak{b}})$ via $D (g \otimes X)= (-1)^{p(g)} g\otimes D(X)$.  Note that 
 $\mathfrak{g} \otimes \mathcal{V}(\bar{\mathfrak{b}})$ is a Lie superalgebra for the bracket
 \begin{equation} \label{new Lie bracket}
  [g \otimes X, h \otimes Y] = (-1)^{(p(X))p(h)}[g,h] \otimes XY, 
\end{equation}
  and can be naturally extended to the semi-direct product $\mathbb{C}[D] \ltimes \left(\mathfrak{g} \otimes \mathcal{V}(\bar{\mathfrak{b}}) \right)$.

 \begin{prop} [\cite{CS2020}] \label{Prop:canonical form}
 Let $V$ be a subspace of $\mathfrak{g}$ such that $[\mathfrak{n},f] \bigoplus V = \mathfrak{b}_+$. 
 \begin{enumerate}[(1)]
 \item There exists a unique even element $N_c \in \mathfrak{n}\otimes \mathcal{V}(\mathfrak{\bar{b}})$ such that 
 $e^{\text{ad} N_c}\mathcal{L}^u = \mathcal{L}^{c}$ and 
 \begin{equation*} 
 Q^c:=\mathcal{L}^{c}-D+f \otimes 1 \in V \otimes \mathcal{V}(\bar{\mathfrak{b}}).
 \end{equation*}
The operator $\mathcal{L}^{c}$  is called the canonical Lax operator associated with our choice of $f$ and $V$.
 \item Let $\{v_i \, |\,  i\in J\}$ be a basis of $V$ and let $Q^c = \sum_{i\in J} v_i \otimes w_i$. Then the W-algebra $\mathcal{W}(\bar{\mathfrak{g}}, f)$ is freely generated by the elements $w_i$'s as a differential algebra.
 \end{enumerate}
 \end{prop}
 
 \begin{proof}
See Theorem 3.14 in  \cite{CS2020}. 
 \end{proof}

There is another realization of the classical SUSY W-algebras. We identify an element in $\mathcal{V}(\bar{\mathfrak{b}})$ with a linear function from the odd space $\mathcal{F}:=(\mathfrak{b}_+\otimes \mathcal{V}(\bar{\mathfrak{b}}))_{\bar{1}}$ to $\mathcal{V}(\bar{\mathfrak{b}})$ defined by 
\begin{equation} \label{eq:b- functional}
    \bar{a}^{(n)}(b\otimes X):= (b|a)X^{(n)}, \quad AB(b\otimes X)= A(b\otimes X) B(b\otimes X) 
\end{equation}
for $a\in \mathfrak{b},$ $b\in \mathfrak{b}_+$, and $A,B,X\in \mathcal{V}(\bar{\mathfrak{b}})$. In particular, for $Q:= \bigoplus_{i\in S} q^i \otimes \bar{q}_i$ and  $A \in \mathcal{V}(\bar{\mathfrak{b}})$, we have 
\begin{equation} \label{eq:Q_univ}
    A(Q)=A.
\end{equation}
For $P\in \mathcal{F}$, we consider the corresponding operator  $\mathcal{L}_P:=  D+P-f\otimes 1$. Two elements $P$ and $P'$ in $\mathcal{F}$ are called gauge equivalent if there is an element $N\in \mathfrak{n}\otimes \mathcal{V}(\bar{\mathfrak{b}})$ such that $e^{\text{ad\,} N}\mathcal{L}_P= \mathcal{L}_{P'}.$ On the other hand, since $\mathcal{W}(\bar{\mathfrak{g}},f)$ is a subalgebra of $\mathcal{V}(\bar{\mathfrak{b}})$, it can also be realized as an algebra of functions on $\mathcal{F}.$ As such, it turns out that this subalgebra $\mathcal{W}(\bar{\mathfrak{g}},f)$ identifies with the gauge invariant functions in the following sense.

\begin{prop} \cite{CS2020}\label{prop:another def of W} 
The differential algebra $\mathcal{W}(\bar{\mathfrak{g}},f)$ is the set of gauge invariant functions $w$, that is $w \in \mathcal{W}(\bar{\mathfrak{g}},f)$ if and only if 
 $w(P)=w(P')$ whenever two elements $P$ and $ P'$ in $\mathcal{F}$ are gauge equivalent.
\end{prop}

 For $a \in \mathcal{V}(\bar{\mathfrak{b}})$, we denote its variational derivative with respect to $\bar{\mathfrak{b}}$ by
\begin{equation} \label{eq:varational_matrix}
\frac{\delta a}{\delta q} := \sum_{i\in S} q_i \otimes \frac{\delta a}{\delta \bar{q}_i}. 
\end{equation}
The following theorem holds.
\begin{thm}  \cite{CS2020}
The degree 1 Lie superalgebra bracket on $\int \mathcal{W}(\bar{\mathfrak{g}}, f)$ can be realized via the universal Lax operator and variational derivatives. More precisely, for $w_1, w_2 \in \mathcal{W} (\bar{\mathfrak{g}}, f)$, 
 \[  \{\smallint w_1\, , \  \smallint w_2 \} = -  \int \left( \left. \frac{\delta w_1 }{\delta q} \right| \left[\mathcal{L}^u,  \frac{\delta w_2 }{\delta q} \right]\right). \]
\end{thm}

\section{SUSY quadratic Gelfand-Dickey brackets and SUSY W-algebras} \label{sec:GD bracket}

\subsection{SUSY pseudo-differential operators}
\label{subsec:operator} Let $N$ be an integer greater than $2$.
Let $\mathcal{W}_N:=\mathbb{C}[u_i^{(n)}|i=1,\cdots, N, \, n\in \mathbb{Z}_+]$ be the differential superalgebra where the generators have parity $p(u_i)\equiv i$ (mod $2$) and 
the odd derivation $D$ is defined on the generators by $D(u_i^{(n)})=u_i^{(n+1)}$. The algebra of pseudo-differential operators on $\mathcal{W}_N$ is the superspace of Laurent series 
  $\mathcal{W}_N(\!(D^{-1})\!):= \mathcal{W}_N \otimes \mathbb{C}(\!( D^{-1})\!)$ endowed with the product 
\begin{equation}
    \begin{aligned}
         & D\cdot aD^m:= a'D^m +(-1)^{p(a)}a D^{m+1},\\
         & D^{-1}\cdot a D^m = \sum_{k\in \mathbb{Z}_+}(-1)^{\frac{k(k-1)}{2}+(k-1)p(a)}a^{(k)}D^{m-k-1},
    \end{aligned}
\end{equation}
where $a\in \mathcal{W}_N$, $m\in \mathbb{Z}$ and $a':= D(a), \, a^{(n)} := D^n(a)$. 
In general, for $n\in \mathbb{Z}$, we have
\begin{equation}
    \begin{aligned}
        & D^{2n}\cdot a D^m=\sum_{j\in\mathbb{Z}_+}{ \binom{n}{j} } a^{(2j)}D^{m+2n-2j},\\
        & D^{2n+1}\cdot a D^m=\sum_{j\in\mathbb{Z}_+}{ \binom{n}{j} } \Big( (-1)^{p(a)}a^{(2j)}D^{m+1+2n-2j}+a^{(2j+1)}D^{m+2n-2j} \Big).\\
    \end{aligned}
\end{equation}

Elements in the subalgebra $\mathcal{W}_N[D]:= \mathcal{W}_N \otimes \mathbb{C}[D]$  are called  super-differential operators.
 For a pseudo-differential operator $A(D)=\sum_{m\in \mathbb{Z}}a_m D^m$, its adjoint $A^*(D)\in\mathcal{W}_N(\!(D^{-1})\!)$ is defined by 
\begin{equation}
    \Big(\sum_{m\in \mathbb{Z}_+} aD^m\Big)^*:=\sum_{m\in \mathbb{Z}_+} (-1)^{mp(a)+\frac{m(m+1)}{2}}D^m a.
\end{equation}
Thus  
$ D^*=-D$,  $(D^2)^*=-D^2$, $ (D^{-1})^*=D^{-1}$ and $a^*=a$
for any $a\in \mathcal{W}_N$. We have the following lemma.
\begin{lem} Let $A$ and $B$ be pseudo-differential operators on $\mathcal{W}_N$. Then
\begin{equation} \label{eq:adjoing comm}
(AB)^*=(-1)^{p(A)p(B)}B^*A^*.
\end{equation}
\end{lem}
\begin{proof}
First, one can easily check \eqref{eq:adjoing comm} is true when $A=D^m$ and $B=D^n$ for any $m,n\in \mathbb{Z}$ and when $A\in \mathcal{W}_N.$ Now suppose that $A=D$ and $B=bD^n$ for some $b\in \mathcal{W}_N$ and $n\in \mathbb{Z}.$ Then 
    \begin{equation}
        \begin{aligned}
           &  (AB)^*= (D\cdot bD^n)^* = (b'D^n+(-1)^{p(b)}b D^{n+1})^* \\
            & = (-1)^{n(p(b)+1)+\frac{n(n+1)}{2}}D^nb'+(-1)^{np(b)+\frac{(n+1)(n+2)}{2}}D^{n+1}b \\
            & = (-1)^{p(b)+n}\big( (-1)^{np(b)+ \frac{n(n+1)}{2}+1}D^n bD\big)=(-1)^{p(A)p(B)}B^*A^*.
        \end{aligned}
    \end{equation} 
    Here, we used $bD=(-1)^{p(b)+1}b'+(-1)^{p(b)}Db.$ Inductively, we can see that \eqref{eq:adjoing comm} holds when $A\in \mathcal{W}_N[D].$ Moreover, we have
    \begin{equation}
        (D^{-1}B)^*D=(D^{-1}B)^*(-D)^*=(-1)^{p(B)+1}(-DD^{-1}B)^*=(-1)^{p(B)}B^*
    \end{equation}
    and hence $(D^{-1}B)^*= (-1)^{p(B)}B^*D^{-1}=(-1)^{p(B)}B^*(D^{-1})^*$. This implies \eqref{eq:adjoing comm} holds for any pseudo-differential operator $A$. 
\end{proof}

For a pseudo-differential operator $A\in \mathcal{W}_N(\!(D^{-1})\!)$, the differential and integral part of $A(D)$ are denoted by 
$A(D)_+:= \sum_{m\in \mathbb{Z_+}}a_m D^m$ and $A(D)_-:= \sum_{m\in \mathbb{Z}_+}a_{-m-1}D^{-m-1}$
respectively, and the residue $ \text{ Res } (A(D))$ of $A(D)$ is the $(-1)$-th coefficient $a_{-1}.$ The properties in the following lemma are useful throughout this paper. 

\begin{lem} \label{lem:res}
Let $A$ and $B$ be two pseudo-differential operators on $\mathcal{W}_N$ and $a\in \mathcal{W}_N$. Then the following properties hold.
\begin{enumerate}[(1)]
    \item $\text{Res }(A)= \text{Res }(A^*),$
    \item $\text{Res }(Aa)=(-1)^{p(a)}\text{Res }(A)a,$
    \item $\smallint \text{Res }(AB)=(-1)^{p(A)p(B)}\smallint \text{Res }(BA).$
\end{enumerate}
\end{lem} 
\begin{proof}
    Let $A=\sum_{i\in \mathbb{Z}}a_i D^i$. Then $\text{Res }(A)=\text{ Res }(A^*)=a_{-1}$ and $\text{ Res }(Aa)=(-1)^{p(a)}\text{ Res }(A)a=(-1)^{p(a)}a_{-1}a.$ Hence we get (1) and (2). Let us show (3). For $m,n\in \mathbb{Z},$ one can observe that $\text{ Res }(aD^{2m}bD^{2n})=0=\text{ Res }(bD^{2n}aD^{2m})$. Also, we have
    \begin{equation} \label{eq:res_lem_1}
        \begin{aligned}
             & \text{ Res }(a D^{2m+1} bD^{2n})= (-1)^{p(b)} {\binom{m}{n+m+1} }a  b^{(2m+2n+2)}, \\
            & \text{ Res }(a D^{2m+1} bD^{2n-1})=  { \binom{m}{n+m}  }a  b^{(2m+2n+1)}, \\
        \end{aligned}
    \end{equation}
   and 
       \begin{equation}\label{eq:res_lem_2}
        \begin{aligned}
             &  \text{ Res } (bD^{2n}a D^{2m+1})=  {\binom{n}{n+m+1}  }b a^{(2n+2n+2)},  \\
            &  \text{ Res }( bD^{2n-1}a D^{2m+1})=  {\binom{n-1}{n+m}}b  a^{(2m+2n+1)}, \\
        \end{aligned}
    \end{equation}
    Considering the fact that  ${ \binom{m}{n+m+}  }=(-1)^{n+m} { \binom{n-1}{n+m} }$ and ${\binom{m}{n+m+1}  }=(-1)^{n+m+1} {  \binom{n}{n+m+1}}$ and comparing  \eqref{eq:res_lem_1} and \eqref{eq:res_lem_2} conclude the proof.
\end{proof}

\subsection{SUSY quadratic Gelfand-Dickey bracket} \label{sec:G-D}

Consider the monic super-differential operator  
\begin{equation}
L_N(D)=D^N+u_1D^{N-1}+...+u_N. 
\end{equation}
Following \cite{HN91} and \cite{FFR92} we define the  \textit{SUSY quadratic Gelfand-Dickey bracket} on  $\smallint \mathcal{W}_N$ by 
\begin{equation}\label{SGD}
  \{\smallint a,\smallint b\} := (-1)^{p(a)+N} \text{ Res }\int  \, \Big((L \frac{\delta a}{\delta L})_+L-L( \frac{\delta a}{\delta L}L)_+ \Big) \, \frac{\delta b}{\delta L},
\end{equation}
where for all $\smallint a \in \smallint \mathcal{W}_N$,
the variational derivative of $ \smallint a$ with respect to $L$  is the pseudo-differential operator 
\begin{equation} \label{Xvder}
    \frac{\delta a}{\delta L} := (-1)^{p(a)}\sum_{k=1}^N D^{k-N-1}\sum_{l \in \mathbb{Z}_+}(-1)^{k+kl+\frac{l(l+1)}{2}}D^l(\frac{\partial a}{\partial u_{k}^{(l)}})\in \mathcal{W}_N(\!(D^{-1})\!).
\end{equation}
Hence the parity of $\frac{\delta a}{\delta L}$ is $p(a)+N+1$ (mod $2$).
We remark that, at this moment, it is not clear if the bracket \eqref{SGD} is a degree $1$ Lie superalgebra bracket or not. We will prove this by realizing it as a reduction of a affine SUSY PVA bracket. Let us consider the odd linear matrix differential operator $\mathcal{L}^{can}$ in  $\mathfrak{gl}(m|n)\otimes \mathcal{W}_N[D]$, where  $m=n+1,\, n$ and $N=m+n$: 
\begin{equation} \label{eq:L}
\begin{aligned}
   \mathcal{L}^{can} 
   & =\sum_{i=1}^N  e_{ii}\otimes D +\sum_{i=1}^N 
 e_{iN}\otimes (-1)^{N+1}u_{N+1-i}-\sum_{i=1}^{N-1}e_{i+1,i}\otimes 1. 
 \end{aligned}
 \end{equation}
Here we use the basis $\{e_{ij}|i,j=1,\cdots, N\}$ of $ \mathfrak{gl}(m|n)$ introduced in Example \ref{ex:gl(n+1|n)} and \ref{ex:gl(n|n)} so that $p(e_{ij})\equiv i+j$ (mod $2$). In the space of matrix operators $\mathfrak{gl}(m|n) \otimes \mathcal{W}_N[D],$ we consider the associative multiplication and Lie bracket defined by 
\begin{equation} \label{eq:bracket_2}
  \begin{aligned}
      & (e_{ij}\otimes a)\cdot (e_{kl}\otimes b)=(-1)^{(k+l)p(a)} \delta_{jk} e_{il}\otimes ab, \\
& [e_{ij}\otimes a, e_{kl}\otimes b]=(e_{ij}\otimes a)\cdot (e_{kl}\otimes b) - (-1)^{(i+j+p(a))(k+l+p(b))}(e_{kl}\otimes b)\cdot (e_{ij}\otimes a) .
  \end{aligned}  
\end{equation}

\vskip 5mm

We can visualize $\mathcal{L}^{can}$ in its matrix form
 \begin{equation} \label{eq:matrix_lax_2}
\mathcal{L}^{can}={ \begin{pmatrix} 
   D &0&\dots&\dots &0  &  (-1)^N u_N \\
   -1& D &\dots&\dots &0 &  (-1)^{N-1} u_{N-1} \\
  \vdots &\vdots& \ddots &  & \vdots &  \vdots \\
   \vdots &\vdots& &  \ddots &\vdots  &  \vdots \\
   0&0 & \dots&\dots & D & u_{2}  \\
   0&0 & \dots& \dots &-1  & D-u_{1}  \\
   \end{pmatrix}  =\begin{pmatrix} 
   D &0&\dots&\dots &0  &  s_N (b_{N}-Db_{N-1}) \\ 
    -1& D &\dots&\dots &0 &  s_{N-1} (b_{N-1}-Db_{N-2}) \\
   \vdots &\vdots& \ddots &  & \vdots &  \vdots \\
   \vdots &\vdots& &  \ddots &\vdots  &  \vdots \\
   0&0 & \dots&\dots &D &   s_2(b_2-Db_1)  \\
  0&0 & \dots& \dots &-1  & s_1 b_1  \\
   \end{pmatrix},}
\end{equation}
where $i=0,...,N$,  $s_i=(-1)^{\frac{i(i-1)}{2}+\frac{N(N+1)}{2}}$ and 
\begin{equation}
    b_i(D)=(D^{i-N}L^*(D))_+ \in \mathcal{W}_N[D].
\end{equation}
In particular, we have $b_0= (-1)^{\frac{N(N+1)}{2}}$ and hence $s_1 b_1= D-s_1(b_1-Db_0)$ and $b_N=L^*(D).$

We associate to any $a \in \mathcal{W}_N$ the $N\times N$ matrix pseudo-differential operator $\nabla_a$ whose $ij$-th entry for $i,j=1,2,\cdots, N$ is
\begin{equation} \label{eq:matrix differential}
    (\nabla_a)_{ij}=(-1)^{(j+1)p(a)+
    (j+1)(i+1)+
    \frac{j(j-1)}{2}+\frac{i(i-1)}{2}}b_{N-i}(\frac{\delta a}{\delta L})^* D^{j-1}.
\end{equation}
Hence the $4\times 4$ left-top corner of $\nabla_a$ is
\begin{equation}
    \begin{pmatrix}
         b_{N-1}(\frac{\delta a}{\delta L})^*&(-1)^{p(a)+1}b_{N-1}(\frac{\delta a}{\delta L})^* D&-b_{N-1}(\frac{\delta a}{\delta L})^*D^2&(-1)^{p(a)}b_{N-1}(\frac{\delta a}{\delta L})^*D^3\\
         -b_{N-2}(\frac{\delta a}{\delta L})^* & (-1)^{p(a)+1}b_{N-2}(\frac{\delta a}{\delta L})^* D&b_{N-2}(\frac{\delta a}{\delta L})^*D^2&(-1)^{p(a)}b_{N-2}(\frac{\delta a}{\delta L})^*D^3\\
          -b_{N-3}(\frac{\delta a}{\delta L})^* &(-1)^{p(a)}b_{N-3}(\frac{\delta a}{\delta L})^* D&b_{N-3}(\frac{\delta a}{\delta L})^*D^2&(-1)^{p(a)+1}b_{N-3}(\frac{\delta a}{\delta L})^*D^3\\
           b_{N-4}(\frac{\delta a}{\delta L})^* &(-1)^{p(a)}b_{N-4}(\frac{\delta a}{\delta L})^* D&-b_{N-4}(\frac{\delta a}{\delta L})^*D^2&(-1)^{p(a)+1}b_{N-4}(\frac{\delta a}{\delta L})^*D^3
    \end{pmatrix}.
\end{equation}

\begin{lem} \label{lem:Lnabla-nablaL}
    Let $a\in \mathcal{W}_N$. For $\mathcal{L}^{can}$ in \eqref{eq:L}, we have the following equalities in $\mathfrak{gl}(m|n)\otimes \mathcal{W}_N(\!(D^{-1})\!):$
\begin{equation} \label{eq:L nabla}
    \mathcal{L}^{can} \nabla_a = \sum_{i=1}^N e_{1i} \otimes (-1)^{\frac{i(i-1)}{2}+(p(a)+1)(i+1)}L^* (\frac{\delta a}{\delta L})^* 
\end{equation}
and
\begin{equation}\label{eq:nabla L}
    \nabla_a \mathcal{L}^{can}=
    \left\{\begin{array}{ll}
         \displaystyle\sum_{i=1}^N e_{iN} \otimes (-1)^{p(a)+1+\frac{i(i+1)}{2}}b_i(\frac{\delta a}{\delta L})^* L^*&  \text{ if } N \text{ is even , }\\ 
         
         \displaystyle\sum_{i=1}^N e_{iN} \otimes (-1)^{\frac{i(i-1)}{2}+1}b_i(\frac{\delta a}{\delta L})^* L^*&   \text{ if } N \text{ is odd    . }
    \end{array}\right.
\end{equation}
\end{lem}
\begin{proof}
     It directly follows from computations.
\end{proof}

Simply speaking, \eqref{eq:L nabla} implies that the only nonzero row of $\mathcal{L}^{can} \nabla_a$ is the first row and the first $4$ entries in that row are  
\begin{equation}
 \begin{pmatrix} 
   L^*(\frac{\delta a}{\delta L})^*&(-1)^{p(a)}L^*(\frac{\delta a}{\delta L})^*D&-L^*(\frac{\delta a}{\delta L})^*D^2& (-1)^{p(a)+1}L^*(\frac{\delta a}{\delta L})^*D^3 
   \end{pmatrix}.
\end{equation}
In a similar way, \eqref{eq:nabla L} says that the only nontrivial column of $\nabla_a \mathcal{L}^{can}$ is the last column and its first 4 entries are as follows:
\begin{equation}
    (-1)^{p(a)}\begin{pmatrix} 
     b_1 (\frac{\delta a}{\delta L})^*L^* \\
     b_2 (\frac{\delta a}{\delta L})^*L^* \\
     -b_3 (\frac{\delta a}{\delta L})^*L^*  \\
    -b_4 (\frac{\delta a}{\delta L})^*L^*  \\
   \end{pmatrix} \text{ when $N$ is even}, \quad 
   \begin{pmatrix}
   -b_1 (\frac{\delta a}{\delta L})^*L^* \\
    b_2 (\frac{\delta a}{\delta L})^*L^* \\
     b_3 (\frac{\delta a}{\delta L})^*L^*  \\
   -b_4 (\frac{\delta a}{\delta L})^*L^*  \\
   \end{pmatrix}
   \text{ when $N$ is odd}.
\end{equation}

\begin{lem} \label{lem:2.2-main1}
    Let $a,b$ be two elements in $\mathcal{W}_N$. We have 
    \begin{equation}\label{eq:lem1_state}
        \int  \text{ Res } ((L \frac{\delta a}{\delta L})_+L \frac{\delta b}{\delta L} )=(-1)^{(p(a)+p(b))N} \int \sum_{i=1}^N \text{ Res }(( \frac{\delta a}{\delta L})^*D^{i-1}) \text{ Res }(b_{N-i}( \frac{\delta a}{\delta L})^*L^* ) \, .
    \end{equation}
\end{lem}
\begin{proof}
 For  $j\in \mathbb{Z}$ and $v\in \mathcal{W}_N$, the pseudo-differential operator $Z=D^{-j} v$ satisfies $\text{ Res } (D^{i-1}Z)=\delta_{ij}(-1)^{p(v)}v$.
Hence for any pseudo-differential operator $Z$, we have 
\begin{equation} \label{eq:lem1-1}
 \sum_{i=1}^N (-1)^{i}(LD^{-i})_+  \text{ Res } (D^{i-1} Z)=(-1)^{p(Z)} (LZ_-)_+.
\end{equation}
   In particular, if we substitute $Z=Z_-=\frac{\delta a}{\delta L}$ then
\begin{equation} \label{eq:lem1_Xf}
  \sum_{i=1}^N (-1)^{i}(LD^{-i})_+ \text{ Res } (D^{i-1} \frac{\delta a}{\delta L})=(-1)^{N+p(a)+1} (L \frac{\delta a}{\delta L})_+.
\end{equation}
After multiplying on the left both sides of \eqref{eq:lem1_Xf} by $L \frac{\delta b}{\delta L}$ and taking residues, we get 
$$ \sum_{i=1}^N \text{ Res }( L \frac{\delta b}{\delta L}(LD^{-i})_+) \text{ Res } (D^{i-1} \frac{\delta a}{\delta L})= \text{ Res } ( L \frac{\delta b}{\delta L} (L \frac{\delta a}{\delta L})_+)$$  
using (2) of Lemma \ref{lem:res}.
Finally, by 
(1) and (3) of Lemma \ref{lem:res}, we have 
$$ \int \sum_{i=1}^N \text{ Res }(( \frac{\delta a}{\delta L})^*D^{i-1}) \text{ Res }((D^{-i}L^*)_+(\frac{\delta b}{\delta L})^*L^* )=(-1)^{(p(a)+p(b))N}\int \text{ Res } ((L \frac{\delta a}{\delta L})_+L \frac{\delta b}{\delta L})$$
for any pseudo-differential operators $A$ and $B$,  which is the same equality as \eqref{eq:lem1_state}. 
Here we used $(D^{i-1})^*=(-1)^{\frac{(i-1)i}{2}}(D^{i-1})$, $(D^{-i})^*=(-1)^{\frac{(-i)(-i+1)}{2}}D^{-i}$ and $\frac{(i-1)i}{2}\equiv \frac{(-i)(-i+1)}{2}$ (mod 2).
 \end{proof}

\begin{lem}\label{lem:2.2-main2}
    Let $a,b$ be two elements in $\mathcal{W}_N$. We have 
    \begin{equation}
        \int  \text{ Res } (L(\frac{\delta a}{\delta L} L)_+ \frac{\delta b}{\delta L} )=(-1)^{p(b)+N+1} \int \sum_{i=1}^N (-1)^{i(p(a)+p(b))} \text{ Res }(b_{N-i}(\frac{\delta a}{\delta L})^*)  \text{ Res }(L^*(\frac{\delta b}{\delta L})^*D^{i-1}).
    \end{equation}
   
\end{lem}
\begin{proof}
For pseudo-differential operators $A$ and $B$, we know $\text{ Res }(A_+B_+)=\text{ Res }(A_-B_-)=0$. Hence 
\begin{equation}
    \text{ Res }(AB_+)=\text{ Res }(A_-B_+)=\text{ Res }(A_-B). 
\end{equation}
Also, recall that the differential part of $\frac{\delta a}{\delta L}$ is trivial. Hence 
  \begin{equation} \label{eq:lem2-2}
      \begin{split}
         \int  \text{ Res }( L(\frac{\delta a}{\delta L} L )_+ \frac{\delta b}{\delta L}) &= \int (-1)^{(p(a)+N+1)(p(b)+N+1)} \text{ Res } ( \frac{\delta b}{\delta L} L (\frac{\delta a}{\delta L} L)_+)  \\
          &= \int (-1)^{N(p(a)+p(b)+1)} \text{ Res }( \frac{\delta a}{\delta L} L  (\frac{\delta b}{\delta L} L)_{-}) \\
          & = \int (-1)^{N(p(a)+p(b)+1)} \text{ Res } (  \frac{\delta a}{\delta L} (L  (\frac{\delta b}{\delta L}L)_{-})_+).
      \end{split}
  \end{equation}
 Applying \eqref{eq:lem1-1}  in the preceding proof for $Z=\frac{\delta b}{\delta L} L$, we have 
 \begin{equation} \label{eq:lem2-3-1}
     (-1)^{p(b)+1}(L(\frac{\delta b}{\delta L} L)_-)_+= \sum_{i=1}^{N}(-1)^i (LD^{-i})_+ \text{ Res } (D^{i-1} \frac{\delta b}{\delta L} L).
 \end{equation}
 Multiplying on the left both sides of \eqref{eq:lem2-3-1} by  $\frac{\delta a}{\delta L}$, taking the residue and applying \eqref{eq:lem2-2}, one gets 
  \begin{equation*}
      \begin{split}
         \int  \text{ Res }( L(\frac{\delta a}{\delta L} L )_+ \frac{\delta b}{\delta L}) &= \sum_{i=1}^N (-1)^{N(p(a)+p(b)+1)+i+p(b)+1} \int \text{ Res } \big( \frac{\delta a}{\delta L} (LD^{-i})_+ \text{ Res } ( D^{i-1} \frac{\delta b}{\delta L} L)  \big ) 
         \\
         &= \sum_{i=1}^N (-1)^{N(p(a)+p(b)+1)} \int \text{ Res } (\frac{\delta a}{\delta L}  (LD^{-i})_+) \text{ Res } ( D^{i-1} \frac{\delta b}{\delta L} L).
         \end{split}
  \end{equation*}
  Here, for the last equality, we used (2) of Lemma \ref{lem:res}.
  Finally, the result follows from Lemma \ref{lem:res} (1) and the fact that $(A_+)^*=(A^*)_+$ for any pseudo-differential operator $A$. 
\end{proof}

\begin{thm}  \label{thm:2.2}
Recall that the integers $m,n$ are uniquely defined by $N=m+n$ and $n \leq m \leq n+1$. Let $(\quad | \quad ): \mathfrak{gl}(m|n) \otimes \mathcal{W}_N \times  \mathfrak{gl}(m|n) \otimes \mathcal{W}_N \to \mathcal{W}_N$ be the bilinear map defined by 
\begin{equation} \label{eq:bilinear}
    (e_{ij}\otimes a|e_{kl}\otimes b)= (-1)^{p(a)(k+l)}(-1)^{i+1}\delta_{il}\delta_{jk} \,  ab
\end{equation}
for $i,j,k,l\in \{1,2,\cdots,N=m+n\}$ and $a,b\in \mathcal{W}_N.$ 
Then the Gelfand-Dickey bracket can be written in terms of the matrix operators $\mathcal{L}^{can}$ and $\nabla_a$. More precisely, we have
    \begin{equation}
    \int (\text{ Res }(\nabla_a)| [\mathcal{L}^{can}, \text{ Res }(\nabla_b)]) =(-1)^{N+p(a)} \int \text{ Res } ((L \frac{\delta a}{\delta L})_+L \frac{\delta b}{\delta L}-L(\frac{\delta a}{\delta L} L)_+\frac{\delta b}{\delta L}),
\end{equation}
where the linear map $\text{ Res }: \mathfrak{gl}(m|n)\otimes \mathcal{W}_N(\!(D^{-1})\!)\to \mathcal{W}_N$ is given by $\text{ Res }(e\otimes A):=e\otimes \text{ Res }(A).$
\end{thm}
\begin{proof}
Recall \eqref{eq:matrix differential} and Lemma \ref{lem:Lnabla-nablaL}. By direct computations, we get  
\begin{equation} \label{eqaux1}
\begin{split}
    (\text{ Res } \nabla_a |\text{ Res }(\nabla_b\mathcal{L}^{can}))&= (-1)^{(p(a)+p(b))(N+1)+N+1} \sum_{i=1}^N \text{ Res }((\frac{\delta a}{\delta L})^*D^{i-1}) \text{ Res }(b_{N-i}(\frac{\delta a}{\delta L})^*L^* ), \\
    (\text{ Res } \nabla_a | \text{ Res }(\mathcal{L}^{can} \nabla_b))&=(-1)^{p(a)+p(b) 
    } \sum_{i=1}^N (-1)^{i(p(a)+p(b))}\text{ Res }(b_{N-i}(\frac{\delta a}{\delta L})^*) \text{ Res }(L^*(\frac{\delta b}{\delta L})^*D^{i-1}).
    \end{split}
\end{equation}
Hence by Lemma \ref{lem:2.2-main1} and \ref{lem:2.2-main2}, we have 
\begin{equation}\label{eqaux2}
        \begin{split}
            \int \text{ Res } ((L\frac{\delta a}{\delta L})_+L \frac{\delta b}{\delta L} )&= (-1)^{p(a)+p(b)+N+1}\int (\text{ Res } \nabla_a |\text{ Res }(\nabla_b\mathcal{L}^{can})), \\
            \int \text{ Res } (L(\frac{\delta a}{\delta L} L)_+\frac{\delta b}{\delta L} )&= (-1)^{p(a)+N+1}\int (\text{ Res } \nabla_a | \text{ Res }(\mathcal{L}^{can} \nabla_b)).
        \end{split}
    \end{equation}
Moreover, if we denote by $(\nabla_b)_{-2}$ the coefficient of $D^{-2}$ in $\nabla_b$, it is clear that 
\begin{equation} \label{eqaux3}
    \begin{split}
        \text{ Res }(\nabla_b \mathcal{L}^{can}) &= - \text{ Res }(\nabla_b )\mathcal{L}^{can} + \text{ Res }(\nabla_b ) D + (\nabla_b)_{-2},\\
        \text{ Res }(\mathcal{L}^{can} \nabla_b ) &=  \mathcal{L}^{can} \text{ Res }(\nabla_b) +(-1)^{p(b)} \text{ Res }(\nabla_b ) D +(-1)^{p(b)} (\nabla_b)_{-2}.
    \end{split}
\end{equation}
It follows from equations \eqref{eqaux2} and \eqref{eqaux3} that
\begin{equation}
    \begin{split}
       \int &(\text{ Res }(\nabla_a)| [\mathcal{L}^{can}, \text{ Res } (\nabla_b)]) =  \int (\text{ Res }(\nabla_a)| \mathcal{L}^{can}\text{ Res }(\nabla_b)+(-1)^{p(b)} \text{ Res }(\nabla_b )\mathcal{L}^{can})  \\
        & \hskip 10mm = \int(\text{ Res }(\nabla_a)| \text{ Res }(\mathcal{L}^{can}\nabla_b)+(-1)^{p(b)+1} \text{ Res }(\nabla_b \mathcal{L}^{can}))\\
        & \hskip 10mm =(-1)^{p(b)+1} \int(\text{ Res }(\nabla_a)|\text{ Res }(\nabla_b \mathcal{L}^{can})) + \int(\text{ Res }(\nabla_a)| \text{ Res }(\mathcal{L}^{can}\nabla_b)) \\
        & \hskip 10mm = (-1)^{N+p(a)} \int \text{ Res } ((L \frac{\delta a}{\delta L})_+L\frac{\delta b}{\delta L}-L( \frac{\delta a}{\delta L} L)_+\frac{\delta b}{\delta L}).
    \end{split}
\end{equation}
\end{proof}

\subsection{SUSY W-algebras and SUSY Gelfand-Dickey brackets} \label{subsec:W and GD}
In this section, we consider the classical SUSY W-algebras introduced in Examples \ref{ex:gl(n+1|n)} and \ref{ex:gl(n|n)}. 
The Lie superalgebra $\mathfrak{g}$ is $\mathfrak{gl}(m|n)$ for $m=n+1$ or $m=n$  and is spanned by the elements $e_{ij}$, $i,j\in \{1,2,\cdots, N=m+n\}$ whose parity is $i+j$ (mod $2$). The odd element $f= \sum_{i=1}^{N-1} e_{i+1,i}$ and the subspace $V=\bigoplus_{i=1}^N e_{iN}$ satisfy the decomposition 
\begin{equation} \label{eq:decomp}
\mathfrak{b}_+=[\mathfrak{n},f]\oplus V,
\end{equation}
where $\mathfrak{n}=\bigoplus_{i<j}e_{ij}$ and $\mathfrak{b}_+=\bigoplus_{i\leq j}e_{ij}$. Recall that for any choice of dual bases $q:=\{(q^i,q_i)|i\in S\}$ of $\mathfrak{b}_+$ and $\mathfrak{b}$ we associate the universal Lax operator 
$$ \mathcal{L}^{u}= D+\sum_{i \in S} q^{i} \otimes \bar{q_i} +f \otimes 1\, \,. \,  $$ 
  Proposition \ref{Prop:canonical form} states that there exists a unique operator  $\mathcal{L}^c$ gauge equivalent to $\mathcal{L}^{u}$ of the form 
\begin{equation} \label{eq:canonical lax}
    \mathcal{L}^c= D+ { \begin{pmatrix} 
   0 &0&\dots&\dots &0  &  w_N \\
   -1& 0 &\dots&\dots &0 &  w_{N-1} \\
  \vdots &\vdots& \ddots &  & \vdots &  \vdots \\
   \vdots &\vdots& &  \ddots &\vdots  &  \vdots \\
   0&0 & \dots&\dots & 0 & w_{2}  \\
   0&0 & \dots& \dots &-1  & w_{1}  \\
   \end{pmatrix} }
\end{equation} Moreover, the set $\{w_1 ,w_2, \cdots, w_N \}$ generates $\mathcal{W}(\bar{\mathfrak{g}},f)$. We construct a differential algebra isomorphism from $\mathcal{W}_N$ to $\mathcal{W}(\bar{\mathfrak{g}},f)$ by letting 
\begin{equation} \label{eq:u and w}
    \phi(u_i):= (-1)^{i}w_i \ , \  i=1,2,\cdots, N.
\end{equation}
Note that under this isomorpshim, we can identify $\mathcal{L}^c$ to the matrix $\mathcal{L}^{can}$ in \eqref{eq:matrix_lax_2}. We will denote them both by $\mathcal{L}^c$ and omit $\phi$.

\vskip 3mm

Recall that in Proposition \ref{prop:another def of W}, we used the fact that every element in $a\in \mathcal{V}(\bar{\mathfrak{b}})$ can be considered as a function from the odd space $\mathcal{F}:= (\mathfrak{b}_+\otimes \mathcal{V}(\bar{\mathfrak{b}}))_{\bar{1}}$  to $\mathcal{V}(\bar{\mathfrak{b}})$. Throughout this section, we let 
\[Q:= \sum_{i\in S }q^i\otimes \bar{q}_i.\]
In addition, for $a\in \mathcal{V}(\bar{\mathfrak{b}})$ and a pair $q:=\{(q^i,q_i)|i\in S\}$ of bases of $\mathfrak{b}_+$ and $\mathfrak{b}$, the variational derivative $\frac{\delta a}{\delta q}$ is defined by \eqref{eq:varational_matrix}. The variational derivative can be interpreted as a differential algebra analogue of the gradient in the following sense.

\begin{lem} \label{lem:var_meaning}
Let $a\in \mathcal{V}(\bar{\mathfrak{b}})$, $h\in \mathcal{F}$. For a constant even variable $\epsilon$, we have 
  \[ \int\Big( h \, |\, \frac{\delta{a}}{\delta{q}} \Big) = \int \frac{d}{d\epsilon}a(Q+\epsilon h) |_{\epsilon=0}.\]
\end{lem}
\begin{proof}
    Let us denote $h:=\sum_{i\in S} q^i \otimes h_i$. Since $h$ is an odd element, the parity $p(h_i)=1-\tilde{\imath}$ for $\tilde{\imath}:=p(q^i).$ By \eqref{eq:Q_univ}, we have $a(Q)=a$ and
    hence 
    \[\int \frac{d}{d\epsilon}a(Q+\epsilon h)|_{\epsilon=0}=  \sum_{i\in S,\, n\in \mathbb{Z}_+} \int h_i^{(n)}\frac{\partial a}{\partial \bar{q}_i^{(n)}}=  \sum_{i\in S, \, n\in \mathbb{Z}_+}  \int (-1)^{n(\tilde{\imath}+1)+\frac{n(n+1)}{2}}h_i \Big( \frac{\partial a}{\partial \bar{q}_i^{(n)}}\Big),\]
    for $\tilde{\imath}=p(q_i).$ On the other hand, we have 
    \[  \Big(h\, |\, \frac{\delta a}{\delta q}\Big)=(-1)^{n(\tilde{\imath}+1)+\frac{n(n+1)}{2}}h_i \bigg( \frac{\partial a}{\partial {\bar{q}}_i^{(n)}}\bigg). \]
    This proves the lemma.
\end{proof}

\begin{lem} Let $a$ be an element of $\mathcal{W}_N$ and $Q$, $h$ and $\epsilon $ be the same as in Lemma \ref{lem:var_meaning}. If we consider the function $u[h](\epsilon):=u(Q+\epsilon h)$ where $u=(u_i)_{i=1}^N$ is the sequence of elements in \eqref{eq:u and w}, then we have 
\[\int \Big(h| \frac{\delta a}{\delta q} \Big)=\int \sum_{i=1}^N u_i[h]'(0)\frac{\delta a}{\delta u_i}.\]
\end{lem}
\begin{proof}
By Lemma \ref{lem:var_meaning}, we get 
\begin{equation} \label{eq:lem3.8}
\begin{split}
    \int \Big(h |\frac{\delta a}{\delta q} \Big)
         &=  \int \frac{d}{d\epsilon}a(Q+\epsilon h)|_{\epsilon=0} =  \int \frac{d}{d\epsilon}a(u(Q+\epsilon h))|_{\epsilon=0} \\
        &=  \int \frac{d}{d\epsilon}a(u+\epsilon u[h]'(0)+O(\epsilon^2))|_{\epsilon=0} = \int \sum_{i=1}^N{u_i[h]'(0)\frac{\delta a}{\delta u_i}}.
\end{split}
\end{equation}
Here in $a(Q+\epsilon h)$, $a$ is regarded as a function on $\mathcal{F}$. On the other hand, since $\mathcal{W}_N$ is an  algebra of differential polynomials generated by $N$ elements $u_1,\cdots, u_N$, when we write $w(u(Q+\epsilon h))$, $a$ is regarded as a differential polynomial on $N$ variables.
\end{proof}

For a pseudo-differential operator $M=\sum_{i\in \mathbb{Z}} M_i D^i\in \mathcal{V}(\bar{\mathfrak{b}}_-)(\!(D^{-1})\!)$, we denote
\begin{equation} \label{eq:functional}
    M(h):=M_i(h) D^i
\end{equation}
for $h\in \mathcal{F}.$ Then the following lemma holds.

\begin{lem}
    Let $a$ be in $\mathcal{W}_N$. Then 
    \begin{equation}
        \text{ Res } \, (\frac{\delta a}{\delta L})^*  \, \, L^*(Q+\epsilon h)= (-1)^{p(a)N}\sum_{i=1}^{N} u_i(Q+ \epsilon h) \frac{\delta a}{\delta u_i}.
    \end{equation}
\end{lem}

\begin{proof}
    This result follows from the fact that $ \text{ Res }  \,  (\frac{\delta a}{\delta L})^* \,  L^*=(-1)^{N(p(a)+N+1)} \text{ Res }  \, L \,  \frac{\delta a}{\delta L}$ and the definition of $\frac{\delta a}{\delta L}$.
\end{proof}

For a given $h\in \mathcal{F}$, we consider the following deformation of $\nabla_a$:
\begin{equation} \label{eq:deform_matrix differential}
    (\nabla_a[h](\epsilon) )_{ij}=(-1)^{(j+1)p(a)+
    (j+1)(i+1)+
    \frac{j(j-1)}{2}+\frac{i(i-1)}{2}}b_{N-i}(Q)(\frac{\delta a}{\delta L})^*(Q+\epsilon h),
\end{equation}
where $Q=\sum_{i\in S} q^i\otimes \bar{q}_i\in \mathcal{F}.$

\begin{lem} \label{lem:2.3-4}
    Let  $a$ be in $\mathcal{W}_N$. Then 
    \begin{equation}
       \text{ Res } \, (\nabla_{a} [h] (\epsilon)| \mathcal{L}^c (Q + \epsilon h ))=  (-1)^{p(a)} \sum_{i=1}^{N} u_i(Q + \epsilon h ) \frac{\delta a}{\delta u_i}.
    \end{equation}
\end{lem}
\begin{proof}
   It follows from the fact that for all $w$ in $\mathcal{W}_N$ we have 
    \begin{equation}
        (\nabla_a [h] (\epsilon)| \mathcal{L}^c( Q + \epsilon h))=  (-1)^{p(a)(N+1)} (\frac{\delta a}{\delta L})^*(Q) L^* ( Q+ \epsilon h).
    \end{equation}
\end{proof}

Let $\mathcal{M}$ be a matrix superdifferential operator entries $\mathcal{M}_{ij}\in  \mathcal{V}(\bar{\mathfrak{b}})(\!(D^{-1})\!)$ for $i,j=1,\cdots, N=m+n$. For $h\in \mathcal{F}$, we denote by $\mathcal{M}(h)$ the  matrix differential operator with entries $\mathcal{M}_{ij}(h)$. Here  $\mathcal{M}_{ij}(h)$ is defined in \eqref{eq:functional}. Recall that $\mathfrak{n}\subset \mathfrak{gl}(m|n)$ is the set of positive degree elements and we have defined $N_c\in \mathfrak{n}\otimes \mathcal{V}(\bar{\mathfrak{b}})$ in Proposition \ref{Prop:canonical form}. $N_c$ is fixing the gauge as follows $e^{ad\, N_c}(D+Q- f\otimes 1) = \mathcal{L}^c$.

\begin{lem} \label{eq:2.3lem_main}
For all $a \in \mathcal{W}_N,$
    \begin{equation}
         e^{-ad\, N_c } \text{ Res }\nabla_{a} -  \frac{\delta a}{\delta q} \in \mathfrak{n}\otimes \mathcal{V}(\bar{\mathfrak{b}}).
    \end{equation}
\end{lem}
\begin{proof}
In order to prove the lemma, it is enough to show that 
\begin{equation} \label{eq:lem2.5-1}
    \int (e^{-ad\,N_c} \text{ Res } \nabla_a |h)= \int \Big( \frac{\delta a}{\delta q} \, |\, h \Big)
\end{equation}
for any $h \in \mathcal{F}.$ Using Lemma \ref{lem:var_meaning} and \ref{lem:2.3-4}, the RHS of \eqref{eq:lem2.5-1} can be written as follows:
\begin{equation} \label{eq:lem2.5-2}
      \begin{split}
          \int \Big( \frac{\delta a}{\delta q}\, | \, h \Big) &= (-1)^{p(a)+1} \int \Big( h|\frac{\delta a}{\delta q} \Big)  \\
          &= (-1)^{p(a)+1}\int \sum_{i=1}^N{u_i'(Q+\epsilon h)(0)\frac{\delta f}{\delta u_i}} \\
           &=-\frac{d}{d\epsilon} \int \text{ Res }\big(\nabla_{a}(Q+ \epsilon h) |\mathcal{L}^c(Q+\epsilon h) \big)|_{\epsilon=0}. \\
    \end{split}
\end{equation}
    Now let us show the last term in \eqref{eq:lem2.5-2} also equals to \eqref{eq:lem2.5-1}.
 There is a $\mathfrak{n}$ valued function $M$ such that $N_c(Q+\epsilon h)= N_c(Q)+ \epsilon  M(Q) + O(\epsilon^2)$. For such $M$, we have :
\begin{equation} \label{eq:lem2.5-3}
    \begin{split}
        & \mathcal{L}^c(Q+\epsilon h)=e^{ad N_c(Q +\epsilon h)}(D+Q+\epsilon h-f \otimes 1) =e^{ad N_c+\epsilon ad M +O(\epsilon^2)}(\mathcal{L}(Q)+\epsilon h) \\
        &=\mathcal{L}^c(Q)+ \epsilon \,  e^{ad N_c}(h) +\epsilon \, [M+\frac{1}{2}[N_c,M]+\frac{1}{6}[N_c,[N_c,M]]+\cdots \, ,\mathcal{L}^c(Q)] + O(\epsilon^2)\\
        &=\mathcal{L}^c(Q)+ \epsilon \,  e^{ad N_c}(h) +\epsilon \, [\phi_h,\mathcal{L}^c(Q)] + O(\epsilon^2)
    \end{split}
    \end{equation}
where  $\phi_h=M+\frac{1}{2}[N_c,M]+\frac{1}{6}[N_c,[N_c,M]]+\cdots.$
  Note that since both $\mathcal{L}^c(Q)\nabla_a(Q)$ and $\nabla_a(Q)\mathcal{L}^c(Q)$ are upper triangular, 
    \begin{equation}  \label{eq:lem2.5-4}
        \int  \text{ Res }(\nabla_a (Q) |[\mathcal{M},\mathcal{L}^c(Q)])=0.
    \end{equation}
when $\mathcal{M}$ is an $\mathfrak{n}$-valued function.
Moreover, since the last row of $\nabla_a$ does not vary with $\epsilon$ but only the last column of $\mathcal{L}^c$ is nonzero, for all $\epsilon$,
 we see that
 \begin{equation}  \label{eq:lem2.5-5}
     ( \nabla_a [h](\epsilon) \, |\, \mathcal{L}^c(Q) ) = str ( \nabla_a (Q) \mathcal{L}^c(Q) ),
\end{equation}
and hence 
\begin{equation} \label{eq:lem2.5-6}
    ( \nabla_a [h](0)' | \mathcal{L}^c(Q)) = 0.
\end{equation}
    By \eqref{eq:lem2.5-4} and \eqref{eq:lem2.5-6}, we have 
    \begin{equation}
    \begin{split}
        \frac{d}{d\epsilon} &  \int 
 \text{ Res }(\nabla_a(Q+\epsilon h)  |\mathcal{L}^c(Q+\epsilon h))|_{\epsilon=0}  \\
         & =  \int \text{ Res }( \nabla_a[h](0)'  | \mathcal{L}^c(Q))
        +   \int  \text{ Res } \, (e^{-ad N_c }\nabla_a(Q)| h) + \int  \text{ Res }(\nabla_a(Q) |[\phi_h,\mathcal{L}^c(Q)]) \\
        &=  \int  \text{ Res } \, (e^{-ad N_c }\nabla_a(Q)| h) = - \int  \, (e^{-ad N_c } \text{ Res } \nabla_a| h).
    \end{split}
    \end{equation}
 Hence we proved that 
  \begin{equation}
          \int \Big( \frac{\delta a}{\delta q}\, |\, h\Big)   =  \smallint (e^{-ad N_c }\text{ Res }(\nabla_{a}(Q))| h)
  \end{equation}
for any $h\in \mathcal{F}$.
\end{proof}

\begin{thm} \label{thm:main1}
Let $a,b\in \mathcal{W}_N$ and $\mathcal{L}^u:=D+Q-f\otimes 1$ for $Q=\sum_{i\in I} q^i \otimes \bar{q}_i$. Then 
\begin{equation} \label{redqua}
    \int 
    \Big(\frac{\delta a}{\delta q}\big| \big[\mathcal{L}^u, \frac{\delta b}{\delta q}\big] \Big) =(-1)^{N+p(a)} \int \text{ Res } ((L \frac{\delta a}{\delta L})_+L \frac{\delta b}{\delta L}-L(\frac{\delta a}{\delta L} L)_+ \frac{\delta b}{\delta L}).
\end{equation}    
\end{thm}
\begin{proof}
 Since we have Theorem \ref{thm:2.2}, the following equality directly implies the theorem:
    \begin{equation} \label{eq:main_proof}
        \int 
    \Big(\frac{\delta a}{\delta q}\big| \big[\mathcal{L}^u, \frac{\delta b}{\delta q}\big] \Big) = \int ( \text{ Res }(\nabla_a)| [\mathcal{L}^c, \text{ Res } (\nabla_b)]) \, .
    \end{equation}
Let us show \eqref{eq:main_proof}. For $N_c$ in Lemma \ref{eq:2.3lem_main}, we have 
\begin{equation} \label{eq:2.3thorem_proof1}
         \int(\text{ Res } \nabla_a |[\mathcal{L}^c, \text{ Res } \nabla_b])= \int(e^{-ad\, N_c} \text{ Res } \nabla_a |e^{-ad\, N_c}[\mathcal{L}^c, \text{ Res } \nabla_b]).
\end{equation}
Since $[\mathcal{L}^c, \text{ Res } \nabla_b]$ is an upper triangular matrix by Lemma \ref{lem:Lnabla-nablaL}, we know that
\begin{equation}\label{eq:2.3thorem_proof2}
    \eqref{eq:2.3thorem_proof1}= \int \Big( \frac{\delta a}{\delta q} \, \big| \, \big[e^{-ad\, N_c}\mathcal{L}^c, e^{-ad\, N_c}\nabla_b \big] \Big)= \int \Big(\Big[\frac{\delta a}{\delta q} , \, \mathcal{L}^u \Big] \big| e^{-ad\, N_c}\nabla_b  \Big).
\end{equation}
By Lemma \ref{eq:2.3lem_main}, there is a $\mathfrak{n}$-valued function $M$ such that 
$ e^{-ad\, N_c} \text{ Res } \nabla_a +M= \frac{\delta a}{\delta q}$. Hence 
\begin{equation}
    \Big[\frac{\delta a}{\delta q}, \mathcal{L}^u\Big]=[e^{-ad \, N_c} \text{ Res } \nabla_a+M,\mathcal{L}^u]= e^{-ad \, N_c}[ \text{ Res } \nabla_a,\mathcal{L}^c]+[M, \mathcal{L}^u]
\end{equation}
is also an upper triangular matrix. Therefore using Lemma \ref{eq:2.3lem_main} again, we get 
\begin{equation}
    \eqref{eq:2.3thorem_proof2}=  \int 
    \Big(\frac{\delta a}{\delta q}\big| \big[\mathcal{L}^u, \frac{\delta b}{\delta q}\big] \Big).
\end{equation}
and conclude \eqref{eq:main_proof} holds.
\end{proof}

\begin{lem}
    The $\chi$-bracket between the generators $u_i$'s of $\mathcal{W}_N$ can be equivalently defined by the relation
    \begin{equation} \label{eq:LL}
    \begin{split}
 \{L(z) \, {}_{\chi} \, L(w)\}&= \langle \big(L(D+\chi+w)  (D+\chi+w-z)(D^2-\chi^2+w^2-z^2)
 ^{-1}\big)_+L(D+w) \\
 & -L(D+\chi+w) \big( (D+\chi+w-z)(D^2-\chi^2+w^2-z^2)
 ^{-1} L(D+w)\big)_+ \rangle
 \end{split}
    \end{equation}
  where for any pseudo-differential operator $A= \sum_{n \leq M} a_n D^n$, we denote $a_0$ by $\langle \, A \, \rangle.$  
    Moreover, \eqref{eq:LL} determines the entire $\chi$-bracket on $\mathcal{W}_N$ by \eqref{eq:master}.
   
\end{lem}

\begin{proof}
By equation \eqref{eq:Lie str} we have for all $a \in \mathcal{W}_N$ and all $j=1,..., N$
\begin{equation}
 \{ \smallint a, \, \smallint u_j \}= \sum_{i=1}^N (-1)^{(a+i)j}  \int \{u_i {}_D u_j\}\frac{\delta a}{\delta u_i} \, .
\end{equation}
This implies that for any $a\in \mathcal{W}_N$, 
$$(-1)^{ja}\{u_i{}_D u_j\}(a)= (-1)^{i+N}\Big( [LD^{i-N-1}a]_+L-L[D^{i-N-1}aL]_+ \Big)_{[N-j]}$$
where the subscript $[N-j]$ means the coefficient of $D^{N-j}$.
For any $a,b \in \mathcal{W}_N$, we have $\{a \, {}_{\chi} \, b \} = \{a \, {}_{\chi+D} \, b \}(1)$. Hence
$$(-1)^{i+N} \{u_i \, {}_{\chi} \, u_j \}= \Big([L(\chi+D)(\chi+D)^{i-N-1}]_+L(D)-L(\chi+D)[(\chi+D)^{i-N-1}L(D)]_+ \Big)_{[N-j]}.
$$

Consider two independent odd indeterminate constants $z$ and $w$ 
and extend the domain of $\chi$-bracket to $\mathcal{W}_N[z,w]$ by the Leibniz rule. Thanks to the identities 
$$(D+\chi-z)^{-1}=\sum_{i=0}^{\infty} (-1)^{\frac{i(i+1)}{2}}z^i (D+\chi)^{-i-1}$$
and
 $$\{a z^m {}_\chi b w^n\}:= (-1)^{(p(a)+1)m}z^m \{ a{}_\chi b\}w^n$$
it follows that 
\begin{equation*}
\begin{aligned}
    & \sum_{i,j=1}^N (-1)^{\frac{(N-i)(N-i+3)}{2}}\{u_i z^{N-i} \, {}_{\chi} \, u_j w^{N-j}\}\\
    & = \sum_{i= 1}^N  (-1)^{\frac{(N-i)(N-i+1)}{2}}\Big( (L(D+\chi) z^{N-i} (D+\chi)^{i-N-1})_+L(D)-L(D+\chi) (z^{N-i}(D+\chi)^{i-N-1} L(D))_+\Big) (w)\\
    & = \Big( \big(L(D+\chi)  (D+\chi-z)^{-1}\big)_+L(D)-L(D+\chi) \big((D+\chi-z)^{-1} L(D)\big)_+\Big)(w).
    \end{aligned}
\end{equation*}
The Lemma follows from the identities 
$$ L(z)= \frac{1-\mathbf{i}}{2}\sum_{j=1}^N u_j (-1)^{(N-j)(N-j+3)/2} (\mathbf{i}z)^{N-j} + \frac{1+\mathbf{i}}{2}\sum_{j=1}^N u_j (-1)^{(N-j)(N-j+3)/2} (-\mathbf{i}z)^{N-j} \, $$
and 
$$(D+\chi+\omega+\mathbf{i}z)^2=(D+\chi+\omega-\mathbf{i}z)^2 \, $$
where $\mathbf{i}$ is the imaginary number in $\mathbf{i}^2=-1.$

\end{proof}

\begin{defn} \label{def:quad_GD}
    \ We denote the odd SUSY PVA bracket \eqref{eq:LL} on $\mathcal{W}_N$ by $\{ \, {}_{\chi} \, \}^q$ since it
    induces the quadratic SUSY Gelfand-Dickey bracket on the space of functionals $\smallint \mathcal{W}_N$.
\end{defn}

\subsection{Integrable hierarchy on $\mathcal{W}_{2n}$} \label{subsec:integrability_even}
In this section we consider the case where the super-differential operator 
\begin{equation*}
 L=D^{2n}+u_1 D^{2n-1}+...+u_{2n}
 \end{equation*}
has even degree $2n$ for $n>2$.
The quadratic bracket \eqref{eq:LL} can be deformed after replacing $L$ by $L+ \epsilon$ where $\epsilon$ is a even constant. Explicitly, the following formula defines an odd SUSY PVA bracket on $\mathcal{W}_{2n}$ compatible with the quadratic one, which we will call the \textit{linear odd SUSY Gelfand-Dickey bracket }
 \begin{equation} \label{eq:LLlin}
    \begin{split}
 \{L(z) \, {}_{\chi} \, L(w)\}^o&= \langle \,  L(D+\chi+w)  (D+\chi+w-z)(D^2-\chi^2+w^2-z^2)
 ^{-1} \, \rangle \\
 & - \langle \, (D+\chi+w-z)(D^2-\chi^2+w^2-z^2)
 ^{-1} L(D+w) \,  \rangle.
 \end{split}
    \end{equation}
Moreover, Theorem \ref{thm:main1} implies that its reduced form to the quotient $ \smallint \mathcal{W}_{2n}$ is
\begin{equation} \label{linred}
    \smallint 
    \{ \, a \, {}_{\chi} \, b  \,\}^o |_{ \chi=0} =(-1)^{p(a)} \int \text{ Res } \, (L \frac{\delta a}{\delta L}-\frac{\delta a}{\delta L} L) \frac{\delta b}{\delta L}.
\end{equation} 
Recall that given an odd SUSY PVA bracket on a differential algebra $\mathcal{P}$ and a functional $\smallint a \in \smallint \mathcal{P}$, the map $\{ \, \smallint a \, {}_{\chi} \, \cdot\, \}|_{ \chi=0}$ defines a derivation of $\mathcal{P}$ of parity $p(a)+1$. This derivation is evolutionary, which means that it supercommutes with $D$. Moreover, if $\mathcal{P}$ is an algebra of differential polynomials generated by elements $\{u_i\}_{i \in I}$ and $d$ is an even or odd evolutionary derivation, we have for all $b \in \mathcal{V}$
$$ \smallint d(b) = \sum_{i \in I} \int \, d(u_i) \frac{\delta b}{\delta u_i} \, .$$
Hence we deduce from \eqref{linred} and \eqref{redqua} that for all $ \smallint a \in \smallint \mathcal{W}_{2n}$
\begin{equation} \label{laxequations}
    \begin{split}
        & \{ \smallint \, a \, {}_{\chi} \, L \, \}^q|_{\chi=0}  = (-1)^{p(a)} ((L \frac{\delta a}{\delta L})_+L-L(\frac{\delta a}{\delta L} L)_+ ) \, ,\\
       &  {\{ \smallint \, a \, {}_{\chi} \, L \, \}^o}|_{\chi=0} = (-1)^{p(a)}[ \, L \, , \,  \frac{\delta a}{\delta L} \, ]_{+} \, , 
    \end{split}
\end{equation}
where  $\{ \, \smallint \, a \, {}_{\chi} \, L \, \}:=\sum_{i=1}^{2n}\{ \, \smallint\, a\, {}_{\chi} \, u_i\, \} D^{2n-i}$ and $[\, A \, , \, B \, ]$ denotes the commutator of two operators $A$ and $B$.

\vspace{5 mm}
We now relate both the quadratic and odd linear Gelfand-Dickey brackets  to a hierarchy of evolutionary derivations on the differential algebra $\mathcal{W}_{2n}$. There exists a unique $n$-th root of $L$ denoted by $L^{1/n}$ with  leading term $D^2$. Note that the $2n$-th root of $L$ does not exist in $\mathcal{W}_{2n}((D^{-1}))$.  Using $L^{1/n}$ we define the even derivations $(d/dt_k)_{k \geq 1}$ of $\mathcal{W}_{2n}$ in the classical way
\begin{equation} \label{evenhie}
    \frac{dL}{dt_k}=[\, L^{k/n}_+ \, , \, L \, ] \, , \, \, \, \, \, \, \, \, \, [ \, d/dt_k \, , \, D \, ]=0 \, .
\end{equation}
Due to the fact that $L$ is even, the nonSUSY picture translates directly to our setting.  Even though these arguments are very standard \cite{Dic91}, we rewrite them here for the unfamiliar reader. 
\begin{lem} \label{lem:density}
Let $k$ and $l$ be positive integers. 
    \begin{enumerate}
        \item [$(1)$] The derivations $d/dt_k$ and $d/dt_l$  commute,
        \item [$(2)$] The residue of $L^{l/n}$ is conserved for $d/dt_k$, i.e.
        $ \int \frac{d}{dt_k} \, \textit{ Res } \, L^{l/n} = 0 .$
    \end{enumerate}
\end{lem}
\begin{proof} 
    Let $X=\frac{d }{dt_k}(L^{1/n})-[\, L^{k/n}_+ \, , \, L^{1/n} \, ]$. We have 
    $$ \sum_{j=0}^{n-1} L^{j/n}X L^{(n-j-1)/n}= \frac{dL}{dt_k}-[\, L^{k/n}_+ \, , \, L \, ] = 0 \, ,$$
    hence $X=0$ by considering the leading term of the LHS. It follows by the Leibniz rule that
    $$\frac{d }{dt_k}(L^{l/n})=[\, L^{k/n}_+ \, , \, L^{l/n} \, ] \,  \, \, \text{   for all   } \, l \geq 1 \, .$$  \\
    From that equation and the fact that the residue of a commutator is always a total derivative, we get $(2)$.  
    As for the first statement,
    \begin{equation}
        \begin{split}
            \frac{d^2 L}{dt_k dt_l}-\frac{d^2 L}{dt_l dt_k} &= [ \, [\, L^{k/n}_+ \, , \, L^{l/n} \, ]_+ ,\, L \,] + [ \, L^{l/n}_+ \, [\, L^{k/n}_+ \, , \, L \, ] \, ]\\
            &- [ \, [\, L^{l/n}_+ \, , \, L^{k/n} \, ]_+ ,\, L \,] - [ \, L^{k/n}_+ \, [\, L^{l/n}_+ \, , \, L \, ] \, ] \\
            &= [ \, [\, L^{k/n}_+ \, , \, L^{l/n} \, ]_+ - [\, L^{k/n}_+ \, , \, L^{l/n}_+ \, ]\, - [\, L^{l/n}_+ \, , \, L^{k/n} \, ]_+ \, ,\, L \,] \\
            &= [ \, [\, L^{k/n}_+ \, , \, L^{l/n}_{-} \, ]_+ + [\, L^{k/n}_+ \, , \, L^{l/n}_+ \, ]_+\, + [\, L^{k/n}_{-} \, , \, L^{l/n}_{+} \, ]_+ \, ,\, L \,]  \\
            &= [ \, [\, L^{k/n} \, , \, L^{l/n} \, ]_+ - [\, L^{k/n}_{-} \, , \, L^{l/n}_{-} \, ]_+  \, ,\, L \,] =0 \, .
        \end{split}
    \end{equation}
\end{proof}
Now we compute the variational derivatives of the residues of fractional powers of $L$. For any integer $k \geq 1$ we define the odd functionals
\begin{equation}
    \int \, h_k = -\frac{n}{k}\, \int \, \text{res} \, L^{k/n} \, .
\end{equation}
\begin{lem} \label{varderhk}
For any integer $k \geq 1$, we have  
    \begin{equation}
       \frac{\delta h_k}{\delta L}= -(L^{k/n-1})_{-} \, \, .
    \end{equation}
\end{lem}
\begin{proof}
Recall that for any $ \smallint a \in \smallint \mathcal{W}_{2n}$, the variational derivative $\frac{\delta a}{\delta L}$ is uniquely defined by the universal property
\begin{equation} \label{univpro}
\int \, a(u + \epsilon h) = \int \, a(u) \, + \epsilon \int \, \text{ Res } (\sum_{i=1}^{2n} h_i D^{2n-i}) \frac{\delta a}{\delta L} + \text{o}(\epsilon) \, \, ,
\end{equation}
where $(h_i)_{i=1,...,2n} \in (\mathcal{W}_{2n})^{2n}$ is any $2n$-tuple satisfying $p(h_i)=p(u_i)$ and $\epsilon$ is a even constant parameter. In the formula above, differential polynomials are interpretated as functions of the generators $u$.
   Let us fix such a $2n$-tuple $(h_i)_{i=1,...,2n}$ and write $\delta L=\sum_{i=1}^{2n} h_i D^{2n-i}$. We denote $L^{k/n}(u)$ by $X$. There exists a unique pseudo-differential operator $\delta X$ such that 
   $$L^{k/m}(u+\epsilon h)=X+ \epsilon \delta X + o(\epsilon) \, .$$ 
   Taking the differential of both sides in the equality $(L^{k/n})^n=L^k$ we get
    $$ \delta X X^{n-1}+... + X^{n-1} \delta X= \delta L L^{k-1}+...+L^{k-1} \delta L,$$
    hence 
    $$ \delta X +... + X^{n-1} \delta X X^{1-n}= \delta L L^{k-1}X^{1-n}+...+L^{k-1} \delta L X^{1-n}.$$
    Using Lemma \ref{lem:res} we deduce that 
$$ n  \int  \text{ Res } \delta X= k \text{  Res  } \int \delta L L^{k/n-1}$$
which ends the proof since we have
$$  \int \text{ Res } L^{k/n}(u+\epsilon h) - \int \text{ Res } L^{k/n}(u)  \,=\epsilon \int \text{ Res } \delta X+ o(\epsilon) $$
and one can take $a(u)= \text{  Res  } L^{k/n}(u)$ in equation \eqref{univpro}.

\end{proof}

It follows directly from the previous Lemma and equation \eqref{laxequations} that for all $k \geq 1$
\begin{equation}
   \{ \smallint \, h_k \, {}_{\chi} \, L \, \}^q|_{ \chi=0} = [ \, L^{k/n}_+ \, , \, L \, ] \, \ =
    \, {\{   \smallint \, h_{k+n} \, {}_{\chi} \, L \, \}^o}|_{  \chi=0}\, .
\end{equation}
 Hence we have shown that 
 \begin{thm} \label{thm:integrability_even}
     The hierarchy of evolutionary derivations \eqref{evenhie} is biHamiltonian for the two compatible odd SUSY PVA brackets $\{ \, \, {}_{\chi} \, \, \}^q$ and $\{ \, \, {}_{\chi} \, \, \}^o$. Its hamiltonians are the functionals $( \smallint \, h_k)_{k \geq 1}$ which are conserved for each derivation in the hierarchy. They are in involution for the  induced odd linear and quadratic SUSY Gelfand-Dickey brackets.
 \end{thm}

\section{Even SUSY PVAs and $\mathcal{W}_{2n+1}$}
\label{sec:even PVA}
We now consider the case where our algebra of differential polynomials is generated by an odd monic differential operator
$$ L(D)= D^{2n+1}+u_1D^{2n}+... +u_{2n+1}.$$
There exists a unique $(2n+1)$-th monic root of $L$ which enables us to construct the hierarchy of evolutionary even derivations of $\mathcal{W}_{2n+1}$
\begin{equation} \label{evenhierarchy} \frac{dL}{dt_k}=[(L^{\frac{2k}{2n+1}})_+, L] \, , \, \, \,  \, [\frac{d}{dt_k} \, , \, D] =0 \, , \, \, \, k \geq 1.
\end{equation}
 These derivations pairwise commute using the same argument as in section \ref{subsec:integrability_even}. They have common conserved densities given by the residues of fractional powers of $L$. However, those corresponding to even operators are trivial since 
 \begin{equation}\label{evendensity}
    \int  \text{ Res } L^{\frac{2q}{2n+1}}=
 \int  \text{ Res } [ \, L^{\frac{q}{2n+1}} \, , \, L^{\frac{q}{2n+1}} \, ] = 0 \, .
 \end{equation}
Hence, the non trivial conserved densities of the hierarchy are even functionals. To find an underlying Hamiltonian structure one must adapt our language since the usual (odd) SUSY PVA brackets map even functionals to odd derivations.

\subsection{Even SUSY brackets}
\label{subsec:even PVA}

Let $\mathcal{R}$ be a $\mathbb{C}[D]$-module. A bilinear map 
\begin{equation}
    [ \ {}_\chi \ ]: \mathcal{R}\times \mathcal{R} \to  \mathbb{C}[\chi]\otimes \mathcal{R}   
\end{equation}
is called an {\it even} $\chi$-bracket  if it satisfies 
\begin{enumerate}[]
        \item (even parity) $p([a{}_{\chi} b]) = p(a)+p(b),$
        \item (sesquilinearity) $[Da{}_{\chi} b]= -\chi [a{}_\chi b],$ $[a{}_{\chi} Db]= (-1)^{p(a)}(D+\chi)[a{}_\chi b]$
\end{enumerate}
for $a,b\in\mathcal{R}.$

\begin{defn}
    A $\mathbb{C}[D]$-module $\mathcal{R}$ is called an even SUSY LCA if it is endowed with an even $\chi$-bracket 
    satisfying the following properties:
    \begin{enumerate}[]
        \item (skewsymmetry) $[ b{}_\chi a ] = (-1)^{p(a)p(b)+1} [a{}_{-D-\chi} b ]_{\leftarrow},$
        \item (Jacobi identity) $[a{}_{\chi}[b{}_\gamma c]]-(-1)^{p(a)p(b)}[b{}_{\gamma}[a{}_\chi c]]=[[a{}_{\chi}b]_{\chi+\gamma}c]$.
    \end{enumerate}
    In the Jacobi identity, we assume $\chi$ and $\gamma$ commute with the brackets. For example, 
    $[\chi^n A{}_{\gamma} B]:= \chi^n[A{}_\gamma B]$.
\end{defn}

For an even SUSY LCA $\mathcal{R}$, the superspace $\int \mathcal{R}:= \mathcal{R}/\partial \mathcal{R}$ is a Lie superalgebra for the induced bracket
\[ [\smallint a \, \, \smallint b] := \smallint [ a{}_{\chi} b]|_{\chi=0} \, .\]

\begin{defn}
    Let $\mathcal{P}$ be a unital SUSY differential algebra with an odd derivation $D$. An even SUSY LCA $(\mathcal{P}, D, \{\ {}_\chi \ \})$ is called an even SUSY PVA if it satisfies the (right) Leibniz rule:
    \begin{equation}
        \{a{}_{\chi}bc\}= \{a{}_\chi b\} c + (-1)^{p(a)p(b)}b\{a{}_\chi c\} \text{ for  } a,b,c\in \mathcal{P}.
    \end{equation}
\end{defn}

\begin{example}
    If $\mathcal{P}$ is the algebra of polynomials generated by an odd variable $u$ and an odd derivation $D$, then for any constant polynomial $Q \in \mathbb{C}[\chi]$ the following defines a even SUSY PVA structure on $\mathcal{P}$
    $$ \{ u \, {}_{\chi} \, u \} = Q(\chi^4)\ \, .$$
    This follows from Theorem $4.12$.
    We were not able to find another example on this specific differential algebra.
\end{example}

Let us define the following super-differential operator
     \begin{equation} \label{eq:(D+chi)sign}
    \{a \, {}_{D}  \, b \}:= \sum_{m\in \mathbb{Z}_+}(-1)^{m(p(a)+p(b))+\frac{m(m+1)}{2}}a_{[m]}b \, D^m 
     \end{equation}
where $a_{[m]}b$ is the $m$-th coefficient in the $\chi$-bracket $\{a{}_{\chi}b\}:=\sum_{m\in \mathbb{Z}_+}\chi^m a_{[m]}b$. Conversely, one retrieves the even $\chi$-bracket from this super-differential operator via the formula
\[\{a{}_{\chi} b\}= \{a{}_{\chi+D} b\} (1).\]

\begin{lem} \label{lem:ev_Lib_skewsymmetry}
    Let $a,b,c$ be elements in an even SUSY PVA $\mathcal{P}$. Then
     \[ \{a{}_{\chi+D}b\} (c) = (-1)^{p(a)p(b)+1} \{b{}_{-\chi-D}a\}_{\leftarrow} \, c.\]
\end{lem}
\begin{proof}
The statement can be written as an identity between super-differential operators
\begin{equation*}
   (-1)^{p(a)p(b)+1} \sum_{n \in \mathbb{Z}_+} (-\chi-D)^m b_{[m]}a= \sum_{m \in \mathbb{Z}_+} (-1)^{m(a+b)+\frac{m(m+1)}{2}} a_{[m]}b (\chi+D)^m
\end{equation*}
which is equivalent to 
\begin{equation*}
   (-1)^{p(a)p(b)+1} \sum_{m \in \mathbb{Z}_+} (-D)^m b_{[m]}a= \sum_{m \in \mathbb{Z}_+} (-1)^{m(a+b)+\frac{m(m+1)}{2}} a_{[m]}b  \, \, D^m \, .
\end{equation*}
Taking the adjoints of both sides yields 
\begin{equation} \label{eq:lem4.4_proof}
    (-1)^{p(a)p(b)+1} \sum_{m \in \mathbb{Z}_+} (-1)^{m(a+b)+\frac{m(m+1)}{2}} b_{[m]}a \, \,  D^n= \sum_{m \in \mathbb{Z}_+}  (-D)^m a_{[m]}b   \, .
\end{equation}
Note that two differential operators are equal if and only if their symbols are equal. Hence in order to show \eqref{eq:lem4.4_proof} it is enough to see 
\begin{equation*}
\begin{split}
    \sum_{m \in \mathbb{Z}_+}  (-D-\chi)^m (a_{[m]}b) &=
    (-1)^{p(a)p(b)+1} \sum_{m \in \mathbb{Z}_+} (-1)^{m(a+b)+\frac{m(m+1)}{2}} b_{[m]}a \, \,  (D+\chi)^m(1)    \, \\
    &= (-1)^{p(a)p(b)+1} \sum_{n \in \mathbb{Z}_+} (-1)^{m(a+b)+\frac{m(m+1)}{2}} b_{[m]}a \, \, (-1)^{m+\frac{n(n+1)}{2}} \chi^m    \, \\
    &= (-1)^{p(a)p(b)+1} \sum_{m \in \mathbb{Z}_+} \chi^m \,  b_{[m]}a 
\end{split}
\end{equation*}
which holds by the skewsymmetry axiom.
\end{proof}

\begin{prop}[Left Leibniz rule]
    Let $a,b,c$ be elements in an even SUSY PVA $\mathcal{P}.$ Then 
    \[ \{ab{}_{\chi} c\} =(-1)^{p(b)p(c)} \{a{}_{\chi+D}c\} (b) + (-1)^{p(a)(p(b)+p(c))}\{b{}_{\chi+D}c\}( a). \]
\end{prop}

\begin{proof}
    By the Leibniz rule, 
    \[ \{c{}_{\chi} ab\}= (-1)^{p(a)p(b)} \{c{}_{\chi} b\} a+ \{c{}_{\chi} a\} b.\]
    By the skewsymmetry and Lemma \ref{lem:ev_Lib_skewsymmetry}, we get the proposition.   
\end{proof}

\begin{lem}
    Suppose that $\mathcal{P}$ is an even SUSY PVA. Then we have 
    \[ \{a^{(m)}{}_{\chi+D} \, b^{(n)} \} c= (-1)^{m(p(a)+p(b)+n)+ \frac{m(m+1)}{2}+np(a)}(\chi+D)^n \{a {}_{\chi+D} b\} (\chi+D)^m c \]
    for $a,b,c\in \mathcal{P}.$
\end{lem}
\begin{proof} We use \eqref{eq:(D+chi)sign} and the sesquilinearity. More precisely,
    \begin{equation}
    \begin{aligned}
       \{a^{(m)}{}_{\chi+D}b^{(n)}\}c & = (-1)^{m(p(a)+p(b)+n)+\frac{m(m+1)}{2}}\{a{}_{\chi+D} b^{(n)}\}(\chi+D)^m c \\
       & = (-1)^{m(p(a)+p(b)+n)+ \frac{m(m+1)}{2}+np(a)}(\chi+D)^n \{a {}_{\chi+D} b\} (\chi+D)^m c.
    \end{aligned}
    \end{equation}
 
\end{proof}

Consider the freely generated differential algebra 
\begin{equation} \label{eq:evPVA}
    \mathcal{P}:= \mathbb{C}\{ u_i^{(m)}|i\in I, \, m\in \mathbb{Z}\},
\end{equation}
where $I$ is a finite index set and $u_i^{(m)}:= D^m(u_i).$

\begin{thm}[Master Formula]
    Let $\mathcal{P}$ be an even SUSY PVA in \eqref{eq:evPVA} and $a,b\in \mathcal{P}.$ Then 
    \begin{equation} \label{evensusypvqmqster}
    \begin{aligned}
        \{a{}_{\chi} b\} = \sum_{i,j\in I, \, m,n\in \mathbb{Z}_+} & (-1)^{p(a)p(b)+(m+n)\tilde{\imath}+ (p(b)+ \tilde{\imath}+1)(\tilde{\jmath}+n)+\frac{m(m+1)}{2}}\\
        & \frac{\partial b}{\partial u_j^{(n)}} (\chi+D)^n\{u_i{}_{\chi+D} u_j\}(\chi+D)^m  \frac{\partial a}{\partial u_i^{(m)}}.
    \end{aligned}
    \end{equation}
\end{thm}
\begin{proof}
    By the right and left Leibniz rules, 
    \begin{equation}
        \begin{aligned}
            \{a{}_{\chi}b\}  & = \sum_{\substack{i,j\in I\\ m,n\in \mathbb{Z}_+}} (-1)^{(b+\tilde{\jmath}+n)(a+\tilde{\jmath}+n)+(a+\tilde{\imath}+m)(\tilde{\jmath}+n)}
            \frac{\partial b}{\partial u_{j}^{(n)}} \{u_i^{(m)}{}_{\chi+D} u_j^{(n)} \} \frac{\partial a}{\partial u_i^{(m)}} \\
            & = \sum_{\substack{i,j\in I\\ m,n\in \mathbb{Z}_+}} (-1)^{(b+\tilde{\jmath}+n)(a+\tilde{\jmath}+n)+(a+\tilde{\imath}+m)(\tilde{\jmath}+n)
            + m (\tilde{\imath}+\tilde{\jmath}+n)+ \frac{m(m+1)}{2}+n\tilde{\imath}}\\
            & \hskip 2cm
            \frac{\partial b}{\partial u_{j}^{(n)}} (\chi+D)^n\{u_i\, {}_{\chi+D}\, u_j \} (\chi+D)^m\frac{\partial a}{\partial u_i^{(m)}}.
        \end{aligned} 
    \end{equation}
    Simplifying the sign, we get the theorem.
\end{proof}

\begin{cor} \label{coro} Let $a,b\in \mathcal{P}$. Then 
    \[  \{ \smallint a \, , \, \smallint b \}|_{\chi=0}= \sum_{i,j\in I} \int (-1)^{p(a)p(b)+ \tilde{\jmath}(p(b)+1) +\tilde{\imath}\tilde{\jmath}} \frac{\delta b}{\delta u_j}\{u_i{}_D u_j \} \frac{\delta a}{ \delta u_i }.  \]
\end{cor}
\begin{proof}
    Recall $\frac{\delta a}{\delta u_i } = \sum_{m\in \mathbb{Z}_+} (-1)^{m\tilde{\imath}+ \frac{m(m+1)}{2}}$ and 
    \[ \int \frac{\partial b }{\partial u_j^{(n)}}D^n A = \int (-1)^{n(p(b)+\tilde{\jmath})+ \frac{n(n-1)}{2}}D^n \frac{\partial b}{\partial u_j^{(n)}}A\]
    for any $A,b\in \mathcal{P}.$ By the master formula, we get the corollary. 
\end{proof}

\begin{lem} \label{lem:skew}
    For $f(\chi):=\sum_{i\in \mathbb{Z}_+} \chi^i \otimes a_i\in \mathbb{C}[\chi]\otimes \mathcal{P}$, let us denote $f(-\chi-D):= \sum_{i\in \mathbb{Z}_+} (-\chi-D)^i \otimes a_i$.
    Suppose $a,b$ are elements in $\mathcal{P}$ and $g(\chi):= a(\chi+D)^m b\in \mathbb{C}[\chi]\otimes \mathcal{P}.$ Then 
    \[g(-\chi-D)= (-1)^{p(a)p(b)+m(p(a)+p(b))+\frac{m(m+1)}{2}}b(\chi+D)^n a.\]
\end{lem}
\begin{proof} Since $D$ is a derivation, we have
    \[g(-\chi-D)= (-1)^{mp(a)+\frac{m(m+1)}{2}}(\chi+D)^m(a) b= (-1)^{p(a)p(b)+m(p(a)+p(b))+\frac{m(m+1)}{2}}b(\chi+D)^m a. \]
    Here, the sign $\frac{m(m+1)}{2}$ comes from the fact that 
    \begin{equation*}
        \begin{aligned}
            & a(\chi+D)b \xrightarrow{\quad \chi \to(-\chi-D)\quad } (-1)^{p(a)}(-\chi-D)(a) b\\
            & a(\chi+D)^2 b= a(-\chi^2+\partial)b \\
            &  \hskip 1cm\xrightarrow{\quad \chi^2 \to (-\chi-D)^2\quad } -(-\chi-D)^2(a)b +aD^2(b)= -(\chi+D)^2(a)b
            \end{aligned}
    \end{equation*}
     Inductively, we get 
    $$ a(\chi+D)^m b \xrightarrow{\quad \chi \to(-\chi-D)\quad } (-1)^{mp(a)+\frac{m(m+1)}{2}}(\chi+D)^m(a) b.$$
\end{proof}

\begin{prop} \label{prop:even_bracket_skew}
    Let $\mathcal{P}$ be the freely generated differential algebra in \eqref{eq:evPVA} endowed with the even $\chi$-bracket $\{\cdot\, {}_\chi \cdot\, \}$. Suppose the $\chi$-bracket is defined by the master formula \eqref{evensusypvqmqster}. If the skewsymmetry axiom holds on the generators, i.e. if  
    \[ \{u_i {}_{\chi} u_j\} = (-1)^{\tilde{\imath}\tilde{\jmath}+1}\{u_j{}_{-\chi-D} u_i\}_{\leftarrow} \]
    for any $i,j\in I$, then it holds for any $a, b \in \mathcal{P}$.
\end{prop}

\begin{proof}
 Let $a,b\in \mathcal{P}.$ By the master formula and Lemma \ref{lem:skew}, we have 
\begin{equation} \label{eq:proof_skew_1}
\begin{aligned}
    \{a{}_{-\chi-D}b\}_{\leftarrow}= \sum_{i,j\in I, \, m,n\in \mathbb{Z}_+} & (-1)^{p(a)p(b)+(m+n)\tilde{\imath}+ (p(b)+ \tilde{\imath}+1)(\tilde{\jmath}+n)+\frac{m(m+1)}{2}} (-1)^{S_1+S_2+S_3+S_4 +S_5}\\
    & \frac{\partial a}{\partial u_i^{(m)}} (\chi+D)^m\{u_i{}_{-\chi-D} u_j\}(\chi+D)^n  \frac{\partial b}{\partial u_j^{(n)}},
\end{aligned}
\end{equation}
where 
\begin{equation}\label{eq:proof_skew_2}
\begin{aligned}
    & S_1:= (p(a)+\tilde{\imath}+m)(m+\tilde{\imath}+\tilde{\jmath}+p(b)+\tilde{\jmath}),\\
    & S_2:= m(\tilde{\imath}+\tilde{\jmath}+p(b)+\tilde{\jmath})+\frac{m(m+1)}{2},\\
    & S_3:= (\tilde{\imath}+\tilde{\jmath})(p(b)+\tilde{\jmath}),\\
    & S_4:= n (p(b)+\tilde{\jmath}+n)+\frac{n(n+1)}{2}. \\
\end{aligned}
\end{equation}
Here $S_1,$ $S_2,$ $S_3$ and $S_4$ appear when we change the positions of $\frac{\partial a}{\partial u_i^{(m)}},$  $(\chi+D)^m$, $(\chi+D)^n$ and $\frac{\partial b}{\partial u_j^{(n)}}$ using Lemma \ref{lem:skew}. Now, by Lemma \ref{lem:ev_Lib_skewsymmetry}, we have 
\begin{equation}\label{eq:proof_skew_3}
 \{u_i{}_{-\chi-D}u_j\}_{\leftarrow}c=(-1)^{\tilde{\imath}\tilde{\jmath}+1} \{u_j{}_{\chi+D}u_i\}c.
\end{equation}
Applying \eqref{eq:proof_skew_2} and \eqref{eq:proof_skew_3} to \eqref{eq:proof_skew_2}, we get
\begin{equation}
\begin{aligned}
     \{a{}_{-\chi-D}b\}_{\leftarrow}& =  \sum_{i,j\in I, \, m,n\in \mathbb{Z}_+} (-1)^{(m+n)\tilde{\jmath}+ (p(a)+ \tilde{\jmath}+1)(\tilde{\imath}+m)+\frac{n(n+1)}{2}+1} \\
     & \hskip 2cm \frac{\partial a}{\partial u_i^{(m)}} (\chi+D)^m\{u_j{}_{\chi+D} u_i\} (\chi+D)^n  \frac{\partial b}{\partial u_j^{(n)}}\\
     & = (-1)^{p(a)p(b)+1}\{b{}_\chi a\}.
\end{aligned}
\end{equation}
\end{proof}

\begin{prop} \label{prop:even_bracket_Jacobi}
    Let $\mathcal{P}$ be the freely generated differential algebra in \eqref{eq:evPVA} endowed with a even $\chi$-bracket $\{\, {}_\chi \, \}$. Suppose that this $\chi$-bracket satisfies the Leibniz rule.
    If the $\chi$-bracket is skew-symmetric and the Jacobi identity holds for any triple of generators, that is for $i,j,k\in I$,
    \[ \{u_i{}_\chi \{ u_j{}_\gamma u_k\} \}- (-1)^{\tilde{\imath}\tilde{\jmath}}\{u_j{}_\gamma \{ u_i {}_{\chi} u_k\}\} =\{\{u_i{}_{\chi} u_j \}_{\chi+\gamma} u_k\},\]
    then the Jacobi identity holds for any elements in $\mathcal{P}$.
\end{prop}

\begin{proof}
    We aim to show that 
    \begin{equation} \label{eq:proof_Jacobi}
        \{a{}_\chi \{ b{}_\gamma c\} \}- (-1)^{p(a)p(b)}\{b {}_\gamma \{ a {}_{\chi} c\}\} =\{\{a{}_{\chi} b \}_{\chi+\gamma} c\}.
    \end{equation} 
    for $a,b,c\in \mathcal{P}.$
    The first term in \eqref{eq:proof_Jacobi} is 
        \begin{align}
            & \{a{}_{\chi}\{b{}_{\gamma} c\}\}   = \sum_{i,j\in I, \, m,n\in \mathbb{Z}_+}  (-1)^{S_1} \Big\{\, a\, {}_{\chi} \, \frac{\partial c}{\partial u_j^{(n)}} (\chi+D)^n\{u_i{}_{\chi+D} u_j\} (\chi+D)^m  \frac{\partial b}{\partial u_i^{(m)}} \, \Big\} \nonumber\\
            & =  \label{eq:proof_Jacobi-1} \sum_{i,j\in I, \, m,n\in \mathbb{Z}_+}  (-1)^{S_1} \bigg[ \Big\{a{}_\chi \frac{\partial c}{\partial u_j^{(n)}}\Big\}(\gamma+D)^n\{u_i {}_{\gamma+D}u_j\} (\gamma+D)^m \frac{\partial b}{\partial u_i^{(m)}} \\
             &  \label{eq:proof_Jacobi-2} \hskip 2.5cm + (-1)^{S_2} \frac{\partial c}{\partial u_j^{(n)}}(\chi+\gamma+D)^n \{u_i{}_{\chi+\gamma+D}u_j\} (\chi+\gamma+D)^m \Big\{ a{}_{\chi} \frac{\partial b}{\partial u_i^{(m)}}\Big\} \\ 
             & \label{eq:proof_Jacobi-3}\hskip 2.5cm +(-1)^{S_3}  \frac{\partial c}{\partial u_j^{(n)}} (\chi+\gamma+D)^n \{ a{}_{\chi+D}\{u_i{}_{\gamma+D} u_j \}\}(\gamma+D)^m \frac{\partial b}{\partial u_i^{(m)}}\bigg]  .
        \end{align}
    where $S_1=p(b)p(c)+(m+n)\tilde{\imath}+ (p(c)+ \tilde{\imath}+1)(\tilde{\jmath}+n)+\frac{m(m+1)}{2}$, $S_2=p(a)(p(c)+m+\tilde{\imath})$ and $S_3=p(a)(p(c)+\tilde{\jmath}).$ We expand \eqref{eq:proof_Jacobi-1}, \eqref{eq:proof_Jacobi-2} and \eqref{eq:proof_Jacobi-3}  using the formula
    \begin{align}
         & \Big\{a{}_\chi \frac{\partial c}{\partial u_j^{(n)}}\Big\}
         = \sum_{k,r\in I, \, l,q\in \mathbb{Z}_+}(-1)^{S_4}\frac{\partial^2 c}{\partial u_r^{(q)} \partial u_j^{(n)}}(\chi+D)^q\{u_k{}_{\chi+D} u_r\} (\chi+D)^l \frac{\partial a}{\partial u_k^{(l)}},\\
         & \Big\{a{}_\chi \frac{\partial b}{\partial u_i^{(m)}}\Big\}
         = \sum_{k,r\in I, \, l,q\in \mathbb{Z}_+}(-1)^{S_5}\frac{\partial^2 b}{\partial u_r^{(q)} \partial u_i^{(m)}}(\chi+D)^q\{u_k{}_{\chi+D} u_r\} (\chi+D)^l \frac{\partial a}{\partial u_k^{(l)}},\\
         & \{a{}_{\chi+D}\{ u_i{}_{\gamma+D} u_j\}\}=\sum_{k\in I, l\in \mathbb{Z}_+}(-1)^{S_6}  \{u_k{}_{\chi+D}\{u_i{}_{\gamma+D}u_j\}\} (\chi+D)^l \frac{\partial a}{\partial u_k^{(l)}},
    \end{align}
    where $S_4=p(a)(p(c)+\tilde{\jmath}+n)+(q+l)\tilde{k}+ ((p(c)+\tilde{\jmath}+n)+ \tilde{k}+1)(\tilde{r}+q)+\frac{l(l+1)}{2}$, $S_4=p(a)(p(b)+\tilde{\imath}+m)+(q+l)\tilde{k}+ ((p(b)+\tilde{\imath}+m)+ \tilde{k}+1)(\tilde{r}+q)+\frac{l(l+1)}{2}$ and $S_6=(\tilde{\imath}+\tilde{\jmath})(p(a)+\tilde{k}+l)+l(\tilde{k}+\tilde{\imath}+\tilde{\jmath})+\frac{l(l+1)}{2}$.
    Then \eqref{eq:proof_Jacobi-1} is
    \begin{equation}
        \begin{aligned}
            \sum_{\substack{i,j,k,r\in I\\m,n,l,q\in \mathbb{Z}_+}}(-1)^{S_7} \frac{\partial^2 c}{\partial u_r^{(q)} \partial u_j^{(n)}}& \Big[ (\chi+D_1)^q\{u_k{}_{\chi+D_1} u_r\}(\chi+D_1)^l \frac{\partial a}{\partial u_k^{(l)}}\Big]_1 \\
            & \hskip 1cm \Big[ (\gamma+D_2)^n\{u_i{}_{\gamma+D_2} u_j\}(\gamma+D_2)^m \frac{\partial b}{\partial u_i^{(m)}}\Big]_2
        \end{aligned}
    \end{equation}
    where $D_1$ and $D_2$ are the derivation $D$ which act in the first and second parentheses $[ \ ]_1$ and $[\ ]_2$, respectively and $
            S_7= (p(a)+p(b))p(c)+p(c)  (\tilde{\jmath}+n+\tilde{r}+q)+\tilde{\imath}(\tilde{\jmath}+m)+\tilde{\jmath}(\tilde{r}+q+1)
             +\tilde{k}(l+\tilde{r})+(\tilde{r}+q)(n+1)+n
             +\frac{m(m+1)}{2}+\frac{l(l+1)}{2}.$
    Similarly, \eqref{eq:proof_Jacobi-2} is
     \begin{equation}
        \begin{aligned}
            \sum_{\substack{i,j,k,r\in I\\m,n,l,q\in \mathbb{Z}_+}} (-1)^{S_8} \frac{\partial c}{\partial u_j^{(n)}}&  (\chi+\gamma+D)^n\{u_i{}_{\chi+\gamma+D} u_j\}(\chi+\gamma+D)^m \frac{\partial^2 b}{\partial u_r^{(q)}\partial u_i^{(m)}} \\
            & \hskip 1cm  (\chi+D)^q\{u_k{}_{\chi+D} u_r\}(\chi+D)^l \frac{\partial a}{\partial u_k^{(l)}}
        \end{aligned}
    \end{equation}
    where $D_1$ and $D_2$ are the derivation $D$ which only act in the first and second parentheses $[ \ ]_1$ and $[\ ]_2$, respectively and $
            S_8= p(a)p(c)+p(a)p(b)+p(b)p(c)+p(c)(\tilde{\jmath}+ n)+ p(b)(\tilde{r}+q)+\tilde{\imath}(\tilde{\jmath}+m+ \tilde{r}+q)+\tilde{\jmath}+ \tilde{k}(\tilde{r}+l)+(\tilde{r}+q)(m+1)+n+\frac{m(m+1)}{2}+\frac{l(l+1)}{2}.$
    Finally, \eqref{eq:proof_Jacobi-3} is
         \begin{equation}
        \begin{aligned}
            \sum_{\substack{i,j,k\in I\\m,n,l\in \mathbb{Z}_+}} (-1)^{S_9} \frac{\partial c}{\partial u_j^{(n)}}\{ u_k{}_{\chi+D}\{u_i{}_{\gamma+D} u_j\}\} \Big[(\chi+D_1)^l\frac{\partial a}{\partial u_k^{(l)}}\Big]_1 \Big[(\gamma+D_2)^m\frac{\partial b}{\partial u_i^{(m)}}\Big]_2,
        \end{aligned}
        \end{equation}
    where $D_1$ and $D_2$ is the derivation $D$ which only act in the parentheses $[\ ]_1$ and $[\ ]_2$ and $S_9=p(a)p(c)+p(b)p(c)+p(a)\tilde{\imath}+p(c)(\tilde{\jmath}+n)+\tilde{\imath}\tilde{\jmath}+\tilde{\imath}m+\tilde{k}l+\tilde{\imath}\tilde{k}+\tilde{\jmath}\tilde{k}+\tilde{\jmath}+n+\frac{m(m+1)}{2}+\frac{l(l+1)}{2}.$

    The second term in \eqref{eq:proof_Jacobi} is 
        \begin{align}
            & \nonumber \{b{}_\gamma\{a{}_\chi c\}\}\\
            & \label{eq:proof_Jacobi-4} = \sum_{\substack{i,j,k,r\in I \\m,n,l,q\in \mathbb{Z}_+}}(-1)^{T_1}\frac{\partial^2 c}{\partial u_j^{(n)}\partial u_r^{(q)}} \Big[ (\gamma+D_1)^n\{u_i{}_{\gamma+D_1} u_j\}(\gamma+D_1)^m \frac{\partial b}{\partial u_i^{(m)}}\Big]_1 \\
            &\nonumber \hskip 6cm \Big[ (\chi+D_2)^q\{u_k{}_{\chi+D_2} u_r\}(\chi+D_2)^l \frac{\partial a}{\partial u_k^{(l)}}\Big]_2\\
            & \label{eq:proof_Jacobi-5} + \sum_{\substack{i,j,k,r\in I \\m,n,l,q\in \mathbb{Z}_+}}
            (-1)^{T_2}\frac{\partial c}{\partial u_j^{(n)}}(\chi+\gamma+D)^n \{u_i{}_{\chi+\gamma+D}u_j\} (\chi+\gamma+D)^m \\
            & \nonumber \hskip 5cm \frac{\partial^2 a}{\partial u_r^{(q)}\partial u_i^{(m)}}(\gamma+D)^q \{u_k{}_{\gamma+D}u_r\}(\gamma+D)^l \frac{\partial b}{\partial u_k^{(l)}}\\
            &\label{eq:proof_Jacobi-6} + \sum_{\substack{i,j,k\in I \\m,n,l\in \mathbb{Z}_+}}
            (-1)^{T_3}\frac{\partial c}{\partial u_j^{(n)}}\{u_i{}_{\gamma+D}\{u_k{}_{\chi+D}u_j\}\} \Big[(\gamma+D_1)^m\frac{\partial b}{\partial u_i^{(m)}}\Big]_1 \Big[(\chi+D_2)^l\frac{\partial a}{\partial u_k^{(l)}}\Big]_2,
        \end{align}
    where $T_1=p(a)p(c)+p(b)p(c)+p(b)(\tilde{r}+q)+p(c)(\tilde{r}+q+\tilde{\jmath}+n)+\tilde{r}(\tilde{k}+\tilde{\jmath}+n+1)+\tilde{\jmath}(\tilde{\imath}+q)+ nq+q+n+\tilde{\jmath}+\tilde{\imath}m+ \tilde{k}l+\frac{m(m+1)}{2}+\frac{n(n+1)}{2},$ $T_2=p(a)p(b)+p(b)p(c)+p(a)p(c)+p(c)(\tilde{\jmath}+n)+p(a)(\tilde{r}+q)+ \tilde{\imath}(\tilde{\jmath}+m \tilde{r}+q)+\tilde{\jmath}+\tilde{k}(\tilde{r}+q)+\tilde{p}m+mq+ n + \tilde{r}+q+\frac{m(m+1)}{2}+\frac{l(l+1)}{2}$
    and $T_3=p(b)p(c)+p(a)p(c)+p(b)\tilde{k}+c\tilde{\jmath}+p(c)n+ \tilde{k}\tilde{\jmath}+\tilde{k}l+\tilde{\jmath}+ \tilde{\imath}\tilde{k}+\tilde{\imath}\tilde{\jmath}+\tilde{\imath}m+n+\frac{m(m+1)}{2}+\frac{l(l+1)}{2}$.
    The last term in \eqref{eq:proof_Jacobi} can be expanded as follows:
    \begin{align}
        & \nonumber \{\{a{}_\chi b\}_{\chi+\gamma} c\} \\
        & \nonumber =\sum_{\substack{i,j\in I\\ m,n\in \mathbb{Z}_+}}(-1)^{\alpha_1}\Big\{ \frac{\partial b}{\partial u_j^{(n)}} (\chi+D)^n \{ u_i{}_{\chi+D} u_j\} (\chi+D)^m \frac{\partial a}{\partial u_i^{(m)}}\, {}_{\chi+\gamma} \, c \Big\}\\
        & \label{eq:proof_Jacobi-7} = \sum_{\substack{i,j\in I\\ m,n\in \mathbb{Z}_+}}(-1)^{\alpha_1+\alpha_2}\Big\{ \frac{\partial b}{\partial u_j^{(n)}}{}_{\chi+\gamma+D} c\Big\} (\chi+D)^n\{u_i{}_{\chi+D}u_j\} (\chi+D)^m \frac{\partial a}{\partial u_i^{(m)}}  \\
        & \label{eq:proof_Jacobi-8} + \sum_{\substack{i,j\in I\\ m,n\in \mathbb{Z}_+}}(-1)^{\alpha_1+\alpha_3} \Big\{ (\chi+D)^m \frac{\partial a}{\partial u_i^{(m)}}\, {}_{\chi+\gamma+D}\,  c\Big\} \{u_i{}_{-\chi-D}u_j\}_{\leftarrow}(\gamma+D)^{n}\frac{\partial b}{\partial u_j^{(n)}}\\
        & \label{eq:proof_Jacobi-9} + \sum_{\substack{i,j\in I\\ m,n\in \mathbb{Z}_+}}(-1)^{\alpha_1+\alpha_4} \{ \{ u_i{}_{\chi+\gamma+D}u_j \}_{\chi+\gamma+D} c\} \Big[ (\chi+D_1)^m \frac{\partial a}{\partial u_i^{(m)}} \Big]_1 \Big[ (\gamma+D_2)^n \frac{\partial a}{\partial u_j^{(n)}} \Big]_2,
    \end{align}
    where $\alpha_1= p(a)p(b)+(m+n)\tilde{\imath}+ (p(b)+ \tilde{\imath}+1)(\tilde{\jmath}+n)+\frac{m(m+1)}{2}$, $\alpha_2=p(c)(p(a)+n+\tilde{\jmath}),$ $\alpha_3=n(p(a)+p(c)+\tilde{\jmath})+\frac{n(n+1)}{2}+(\tilde{\imath}+\tilde{\jmath})(p(a)+p(c)+\tilde{\imath})$ and $\alpha_4=n(p(a)+p(c)+\tilde{\jmath})+\frac{n(n+1)}{2}+(p(a)+\tilde{\imath})p(c)$.

    Using the master formula once more, \eqref{eq:proof_Jacobi-7} is 
    \begin{equation}
    \begin{aligned}
        \sum_{\substack{i,j,k,r\in I\\ m,n,l,q\in \mathbb{Z}_+}} (-1)^{\alpha_5}   \frac{\partial c}{\partial u_k^{(l)}}(\chi+\gamma+D)^l & \{u_r{}_{\chi+\gamma+D} u_k\}(\chi+\gamma+D)^q  \\ \frac{\partial^2 b}{\partial u_j^{(n)} \partial u_r^{(q)}} 
        & (\chi+D)^n \{u_i{}_{\chi+D} u_j \}(\chi+D)^m \frac{\partial a}{\partial u_i^{(m)}},
    \end{aligned}
    \end{equation}
    where $\alpha_5=p(a)p(b)+p(a)p(c)+p(b)p(c)+p(c)(\tilde{k}+l)+p(b)(n+\tilde{\jmath})+\tilde{r}(n+\tilde{\jmath}+\tilde{k}+q) + \tilde{\imath}(m+\tilde{\jmath})+ q (\tilde{\jmath}+n)$.  The next term \eqref{eq:proof_Jacobi-8} is 
    \begin{equation}
    \begin{aligned}
        \sum_{\substack{i,j,k,r\in I\\ m,n,l,q\in \mathbb{Z}_+}} (-1)^{\alpha_6} \frac{\partial c}{\partial u_k^{(l)}}(\chi+\gamma+D)^l & \{u_r{}_{\chi+\gamma+D}u_k\}(\chi+\gamma+D)^q \\
        \frac{\partial^2 a}{\partial u_i^{(m)}u_r^{(q)}}& (\gamma+D)^m \{ u_j{}_{\gamma+D} u_i\}(\gamma+D)^{n} \frac{\partial b}{\partial u_j^{(n)}},
    \end{aligned}
    \end{equation}
    where $\alpha_6=p(a)p(c)+p(b)p(c)+\tilde{r}(\tilde{\imath}+ m+ \tilde{k}+ q)+ \tilde{\jmath}(\tilde{\imath}+ n)+ q(\tilde{\imath}+m)+p(a)(\tilde{\imath}+m)+ p(c)(\tilde{k}+l)+ \tilde{\imath}+\tilde{k}+ m+l+ \frac{q(q+1)}{2}+\frac{n(n+1)}{2}+1$.
    Here, we used \[\frac{\partial}{\partial u_r^{(q)}}\frac{\partial}{\partial u_i^{(m)}}= (-1)^{(\tilde{r}+q)(\tilde{\imath}+m)}\frac{\partial}{\partial u_i^{(m)}}\frac{\partial}{\partial u_r^{(q)}}\]
    The last term 
    \eqref{eq:proof_Jacobi-9} is 
    \begin{equation}
    \begin{aligned}
        \sum_{\substack{i,j,k\in I\\ m,n,l\in \mathbb{Z}_+}} (-1)^{\alpha_7}&  \frac{\partial c}{\partial u_k^{(l)}}(\chi+\gamma+D)^l \{ \{ u_i{}_{\chi+D}u_j \}{}_{\chi+\gamma+D} u_k\} \\
        & \Big[ (\chi+D_1)^m \frac{\partial a}{\partial u_i^{(m)}}\Big]_1 \Big[(\gamma+D_2)^n \frac{\partial b}{\partial u_j^{(n)}} \Big]_2,
    \end{aligned}
    \end{equation}
    where $\alpha_7=p(a)p(c)+p(b)p(c)+p(c)(\tilde{k}+l)+ \tilde{\imath}m + \tilde{\imath}\tilde{k}+\tilde{\imath}\tilde{\jmath}+ \tilde{\jmath} \tilde{k}+\tilde{\jmath}n + \tilde{k}+l + \frac{m(m+1)}{2}+ \frac{n(n+1)}{2}$.
    By tracking all the signs, we get
\begin{equation}
    \begin{array}{ccccc}
        &  \eqref{eq:proof_Jacobi-1}& -& (-1)^{p(a)p(b)} \eqref{eq:proof_Jacobi-4} & =0  \\
        &   \eqref{eq:proof_Jacobi-2} &+ &0 &=  \eqref{eq:proof_Jacobi-7}\\
        &0 & -&(-1)^{p(a)p(b)}\eqref{eq:proof_Jacobi-2}& = \eqref{eq:proof_Jacobi-8}\\
        & \eqref{eq:proof_Jacobi-3}& - &(-1)^{p(a)p(b)}\eqref{eq:proof_Jacobi-6}&=\eqref{eq:proof_Jacobi-9},
    \end{array}
\end{equation}
which implies the Jacobi identity.
\end{proof}

By Proposition \ref{prop:even_bracket_skew} and \ref{prop:even_bracket_Jacobi}, we conclude the following Theorem.

\begin{thm} \label{thm:even_bracket_skew_Jacobi}
    Let $\mathcal{P}$ be the freely generated differential algebra in \eqref{eq:evPVA} endowed with an even $\chi$-bracket $\{\, {}_\chi \, \}$ given by the master formula. If the skewsymmetry axiom and the Jacobi identity hold for the generators then they hold for any elements in $\mathcal{P}$. 
\end{thm}

\subsection{Linear affine even SUSY PVA} \label{subsec:even linear}
In the rest of this paper, we let $\mathfrak{g}= \mathfrak{gl}(n+1|n)$, $N=2n+1$. 
Recall Example \ref{ex:gl(n+1|n)} that $\mathfrak{g}$ is spanned by $\{e_{ij}|i,j\in I\}$, where $p(e_{ij})=i+j.$ Consider the differential superalgebra \[\mathcal{V}(\bar{\mathfrak{g}})=S(\mathbb{C}[D]\otimes \bar{\mathfrak{g}}).\]

\begin{defn}
    The linear affine even SUSY PVA  bracket on the differential algebra $\mathcal{V}(\bar{\mathfrak{g}})$ is given by 
    \begin{equation}
        \{ \bar{e}_{1i} \, {}_{\chi} \,\bar{e}_{i \, 2n+1} \} = (-1)^{i+1} \, 
    \end{equation}
    for $i=1,2,\cdots, 2n+1.$ All other brackets between generators are set as being trivial. 
\end{defn}
It follows from Theorem \ref{thm:even_bracket_skew_Jacobi} that this bracket is well-defined. Moreover, from Corollary \ref{coro} we have for all $a, b \in \mathcal{V}(\bar{\mathfrak{g}})$, 
\begin{equation}
\begin{split}
    \smallint \, \{ a \, {}_{\chi} \, b \}|_{ \chi = 0} &= \int \sum_{i=1}^{2n+1} (-1)^{(a+1)i} \Big( \frac{\delta a}{\delta \bar{e}_{i\, 2n+1}} \frac{\delta b}{\delta \bar{e}_{1i}}+(-1)^{i+1} \frac{\delta a}{\delta \bar{e}_{1i}} \frac{\delta b}{\delta \bar{e}_{i\, 2n+1}} \Big) \\
    &= (-1)^{p(a)+1} \Big ( \frac{\delta a}{\delta q} | (e_{1\, 2n+1} \otimes 1)  \frac{\delta b}{\delta q}+ \frac{\delta b}{\delta q} (e_{1\, 2n+1}\otimes 1) \Big),
\end{split}
\end{equation}
where $\frac{\delta a}{\delta q}= \sum_{i,j\in I} e_{ij}\otimes \frac{\delta a}{\delta \, \bar{e}_{ij}}.$
Finally, by Lemma \ref{eq:2.3lem_main}, we deduce that for all $a,b \in \mathcal{W}_{2n+1}$  
\begin{equation}
    \int \, \{ a \, {}_{\chi} \, b \}|_{ \chi = 0} = (-1)^{p(a)+1} \int \,  \big(\  \text{ Res  } \nabla_a \, | \,  (e_{1 \, 2n+1}\otimes 1)  \text{  Res  } \nabla_b + \text{ Res } \nabla_b  \, (e_{1 \, 2n+1}\otimes 1) \  \big) \, .
\end{equation}

\begin{lem} 
    Let $a,b$ be two elements in $\mathcal{W}_{2n+1}$. We have 
    \begin{equation}\label{eq:lem1_state_1}
        \int \text{ Res } (L \frac{\delta a}{\delta L} \frac{\delta b}{\delta L} )=(-1)^{p(a)} \int \sum_{i=1}^{2n+1} \text{ Res }((\frac{\delta a}{\delta L})^*D^{i-1}) \text{ Res }(b_{N-i}(\frac{\delta b}{\delta L})^*)
    \end{equation}
\end{lem}
\begin{proof}
 For  $j\in \mathbb{Z}$ and $v\in \mathcal{W}_{2n+1}$, the pseudo-differential operator $Z=D^{-j} v$ satisfies $Res(D^{i-1}Z)=\delta_{ij}(-1)^{p(v)}v$.
Hence for any pseudo-differential operator $Z$, we have 
\begin{equation} \label{eq:lem1}
 \sum_{i=1}^{2n+1} (-1)^{i}(LD^{-i})_+ \text{ Res } (D^{i-1} Z)=(-1)^{p(Z)} (LZ_-)_+.
\end{equation}
   In particular, if we substitute $Z=Z_-=\frac{\delta a}{\delta L}$ then
\begin{equation} \label{eq:lem1_Xf2}
  \sum_{i=1}^{2n+1} (-1)^{i}(LD^{-i})_+ \text{ Res } (D^{i-1}\frac{\delta a}{\delta L})=(-1)^{p(a)} (L \frac{\delta a}{\delta L})_+.
\end{equation}
After multiplying by $\frac{\delta b}{\delta L}$ on the left to the both sides of \eqref{eq:lem1_Xf2} and taking residues, we get 
$$ \sum_{i=1}^{2n+1} \text{ Res }( \frac{\delta b}{\delta L}(LD^{-i})_+) \text{ Res } (D^{i-1}\frac{\delta a}{\delta L})= \text{ Res } ( \frac{\delta b}{\delta L} L \frac{\delta a}{\delta L}),$$  
using (2) of Lemma \ref{lem:res}.
Finally, by 
(1) and (3) of Lemma \ref{lem:res}, we have 
$$ \int \sum_{i=1}^{2n+1}  \text{ Res }((\frac{\delta a}{\delta L})^*D^{i-1}) \text{ Res } ((D^{-i}L^*)_+(\frac{\delta b}{\delta L})^* )=(-1)^{p(a)}\int \text{ Res } (L \frac{\delta a}{\delta L} \frac{\delta b}{\delta L} ) \, . $$
 \end{proof}

\begin{lem}\label{lem:2.2-main22}
    Let $a,b$ be two elements in $\mathcal{W}_{2n+1}$. We have 
    \begin{equation}
        \int \text{ Res } (\frac{\delta a}{\delta L} L \frac{\delta b}{\delta L} )=(-1)^{p(a)+p(b)} \int \sum_{i=1}^{2n+1} (-1)^{i(p(a)+p(b)+1)} \text{ Res }(b_{N-i}(\frac{\delta a}{\delta L})^*) \text{ Res } ((\frac{\delta b}{\delta L})^*D^{i-1}).
    \end{equation}
   
\end{lem}
\begin{proof}
 Applying \eqref{eq:lem1} in the preceding proof for $Z=\frac{\delta b}{\delta L}$, we have 
 \begin{equation} \label{eq:lem2-3}
     (-1)^{p(b)}(L \frac{\delta b}{\delta L})_+= \sum_{i=1}^{2n+1}(-1)^i (LD^{-i})_+ \text{ Res } (D^{i-1} \frac{\delta b}{\delta L}).
 \end{equation}
 On \eqref{eq:lem2-3}, by multiplying $\frac{\delta a}{\delta L}$ on the left, taking the residue and applying \eqref{eq:lem2-2}, one gets 
  \begin{equation*}
         \int  \text{ Res } ( \frac{\delta a}{\delta L} L  \frac{\delta b}{\delta L}) =  \sum_{i=1}^{2n+1}  \int \text{ Res } ( \frac{\delta a}{\delta L} (LD^{-i})_+) \text{ Res } ( D^{i-1} \frac{\delta b}{\delta L} ).
  \end{equation*}
  One ends the proof by taking adjoints inside the residues.
\end{proof}

\begin{lem} \label{keylem}
    For all $a,b \in \mathcal{W}_{2n+1}$, 
    \begin{equation}
      \int \, \{ a \, {}_{\chi} \, b \}|_{\chi = 0} =  \int \text{ Res } (L \frac{\delta a}{\delta L} \frac{\delta b}{\delta L} +(-1)^{p(a)+1} \frac{\delta a}{\delta L} L \frac{\delta b}{\delta L}).  
    \end{equation}
\end{lem}
\begin{proof}
It can be checked by direct computation that 
\begin{equation}
    \begin{split}
        \int  \big( \, \text{ Res }( \nabla_a) | (e_{1\, 2n+1}\otimes 1)\, \,  \text{ Res }( \nabla_b) \, \big)  &=\int \sum_{i=1}^{2n+1} (-1)^{i(p(a)+p(b)+1)} \text{ Res }(b_{N-i}(\frac{\delta a}{\delta L})^*) \text{ Res } ((\frac{\delta b}{\delta L})^*D^{i-1}) \, , \\
        \int  \big(\, \text{ Res }( \nabla_a) | \, \text{ Res }( \nabla_b) \, (e_{1\, 2n+1}\otimes 1) \, \big) &= - \int \sum_{i=1}^{2n+1} \text{ Res }((\frac{\delta a}{\delta L})^*D^{i-1}) \text{ Res }(b_{N-i}(\frac{\delta b}{\delta L})^*) \, .
    \end{split}
\end{equation}
    We conclude the proof using equation \eqref{eq11} and the previous two Lemmas.
\end{proof}

\begin{thm} \label{thm:even bracket W_N}
    The subalgebra $\mathcal{W}_{2n+1} \subset \mathcal{V}(\bar{\mathfrak{g}})$ is stabilized by the linear affine even bracket. In particular, the even bracket on $\mathcal{W}_{2n+1}$ satisfies the Jacobi identity.
\end{thm}
\begin{proof}
    Immediately follows from the previous Lemma and the master formula.
\end{proof}

\begin{lem}
   The linear even SUSY PVA bracket reduced to $\mathcal{W}_{2n+1}$ is given explicitly by the formula 
    \begin{equation} \label{eq:el bracket}
    \begin{split}
        \{ L(z) \, {}_{\chi} \, L(w) \} &= \langle L(D+\chi+w) (D+\chi+w-z)(D^2-\chi^2+w^2-z^2)^{-1} \rangle \\
        & + \langle (D+\chi+w-z)(D^2-\chi^2+w^2-z^2)^{-1} L(D+w) \rangle \, .
    \end{split}
\end{equation}
\end{lem}

\begin{proof}
By the master formula, we have 
$$\{u_i{}_D u_j\}(\frac{\delta a}{\delta u_i})= \Big((-1)^{(j+1)(p(a)+i)} LD^{i-2n-2}\frac{\delta a}{\delta u_i} +(-1)^{i+1+j(p(a)+i)}D^{i-2n-2}\frac{\delta a}{\delta u_i}L\Big)_{[2n+1-j]} \, .$$
This implies that for any $a \in \mathcal{W}_{N}$ 
$$\{u_i{}_D u_j\}(a)= \Big((-1)^{jp(a)+p(a)} LD^{i-2n-2}a +(-1)^{i+1+jp(a)}D^{i-2n-2}aL \Big)_{[2n+1-j]} \, .$$
In other words, 

\begin{equation*}
        \{  u_i  \, {}_{\chi} \, u_j \} =  \Big( L(D+\chi)(D+\chi)^{i-2n-2} +  (-1)^{i+1}(D+\chi)^{i-2n-2} L(D) \Big)_{[2n+1-j]} \, .
\end{equation*}
Therefore 

\begin{equation*}
        \begin{split}
        \{  u_i  \, {}_{\chi} \, L(w) \} &=  \langle L(D+\chi+w)(D+\chi+w)^{i-2n-2} \rangle \\
        & + (-1)^{i+1}\langle (D+\chi+w)^{i-2n-2} L(D+w) \rangle \, .
        \end{split}
    \end{equation*}
 Finally
    \begin{equation*}
        \begin{split}
          \sum_{i=1}^{2n+1} \{ u_i (-1)^{(2n+1-i)(2n-i+4)/2} z^{2n+1-i} \, {}_{\chi} \, L(w) \} 
         &= \langle L(D+\chi+w) (D+\chi+w-z)^{-1} \rangle \\
        & + \langle (D+\chi+w-z)^{-1} L(D+w) \rangle \, .
        \end{split}
    \end{equation*}
The result follows from the identity
\begin{equation*}
        \begin{split}
   L(z) & = \frac{1-\mathbf{i}}{2}\sum_{j=1}^{2n+1} u_j (-1)^{(2n+1-j)(2n-j+4)/2} (\mathbf{i}z)^{2n+1-j}\\
  & \hskip 5mm + \frac{1+\mathbf{i}}{2}\sum_{j=1}^{2n+1} u_j (-1)^{(2n+1-j)(2n-j+4)/2} (-\mathbf{i}z)^{2n+1-j} \, 
        \end{split}
    \end{equation*}
where $\mathbf{i}^2=-1$.
\end{proof}

\begin{defn} \label{def:el_GD}
    We denote the even SUSY PVA bracket \eqref{eq:el bracket} on $\mathcal{W}_{2n+1}$ by $\{ \, {}_{\chi} \, \}^e$ and call it the even linear GD bracket.
\end{defn}

\subsection{Hamiltonian hierarchy on $\mathcal{W}_{2n+1}$} \label{sec:integrability_odd}
In this subsection we relate the linear even Gelfand-Dickey bracket to the hierarchy  \eqref{even-hierarchy}.
Recall that the evolutionary derivations $dt_k$ are defined by the Lax equations
\begin{equation} \label{even-hierarchy} \frac{dL}{dt_k}=[(L^{\frac{2k}{2n+1}})_+, L], \quad \, [\frac{d}{dt_k} \, , \, D] =0 \, , \, \, \, k \geq 1 \, .
\end{equation}

\begin{lem} For all $k, l \geq 1$ we have 
    \begin{enumerate}
        \item [$(1)$] The derivations $d/dt_k$ and $d/dt_l$  commute ,
        \item [$(2)$] The residue of $L^{\frac{l}{2n+1}}$ is conserved for $d/dt_k$, i.e.
        $$ \int \frac{d}{dt_k} \, \textit{ Res } \, L^{\frac{l}{2n+1}} = 0 \, .$$
    \end{enumerate}
\end{lem}
\begin{proof}
    Identical to the even case in section \ref{sec:GD bracket}.
\end{proof}

Now we assume $k$ is an odd integer and define the density

\begin{equation} \label{eq:h_k}
   \int \, h_k =  \, \frac{2n+1}{k} \int \text{ Res } L^{\frac{k}{2n+1}} \, .
\end{equation}
\begin{lem}
For all odd integer $k$, we have 
$$ \frac{\delta {h_k}}{\delta L}=(L^{\frac{k}{2n+1}-1})_{-}\, \, .$$
\end{lem}
\begin{proof}
    The proof is identical to Lemma \ref{varderhk}.
\end{proof}

 \begin{thm} \label{thm:integrability odd}
     The hierarchy of even evolutionary derivations \eqref{even-hierarchy} is Hamiltonian for the linear even SUSY PVA bracket $\{ \, \, {}_{\chi} \, \, \}^e$. Its hamiltonians are the functionals \eqref{eq:h_k} which are conserved for each derivation in the hierarchy. They are in involution for the even linear Gelfand-Dickey bracket.
 \end{thm}
\begin{proof}
Recall that for any evolutionary derivation $d$
 of $\mathcal{W}_{2n+1}$ one has for all $a \in \mathcal{W}_{2n+1}$
 $$ \int \, d(a) =  \int \sum_{i=1}^{2n+1} d(u_i) \frac{\delta a}{\delta u_i}\, .$$
 Hence it follows from Lemma \ref{keylem} that for any homogeneous element $ \smallint a \in \smallint \mathcal{W}_{2n+1}$ the derivation $\{ \smallint \,  a \, {}_{\chi} \, \cdot \}^e|_{ \chi =0}$ of $\mathcal{W}_{2n+1}$ is defined by
 $$ \{ a \, {}_{\chi} \, L \}^{e}|_{ \chi =0} =\Big( L \frac{\delta a}{\delta L} + (-1)^{p(a)+1} \frac{\delta a}{\delta L} L \Big)_+ \, .$$
In particular, for all positive odd integer $k$,
    \begin{equation}
        {\{ \, \smallint \, h_k \, {}_{\chi} \, L(w) \, \}^{e}}|_{ \chi =0}=[L^{\frac{k}{2n+1}-1}_+,L](w).
    \end{equation}

\end{proof}

\newpage

 \bibliographystyle{alpha}
 \bibliography{references}

\end{document}